\documentclass[12pt,a4paper,oneside]{scrartcl}

\usepackage{setspace}
\usepackage{graphicx}
\usepackage{rotating}
\usepackage{natbib}
\usepackage{color}
\usepackage{epic}
\usepackage{rotating}
\usepackage{pict2e}
\usepackage[english]{babel}
\usepackage[tbtags]{amsmath}
\usepackage{amsthm}
\usepackage{amssymb}
\usepackage{bm}  
\usepackage{color}
\usepackage{xspace}    
\usepackage{array}    
\usepackage[latin1]{inputenc}
\usepackage{ulem,epsf}
\usepackage{rotating}
\usepackage{pst-all}
\usepackage{subfigure}
\usepackage{appendix}

\usepackage[bookmarks=true]{hyperref}

\newtheorem{lemma}{Lemma}

\newtheorem{proposition}{Proposition}
\theoremstyle{definition}
\newtheorem{definition}{Definition}

\theoremstyle{remark}
\newtheorem{remark}{Remark}

\begin{document}

\begin{center}
{\LARGE Cross-Ownership as a Structural Explanation for Over- and Underestimation of Default Probability}
\end{center}

\bigskip

\newcommand{\footremember}[2]{%
   \footnote{#2}
    \newcounter{#1}
    \setcounter{#1}{\value{footnote}}%
}
\newcommand{\footrecall}[1]{%
    \footnotemark[\value{#1}]%
}

\begin{center}
  { \large Sabine Karl\footremember{bla}{Institute of Mathematics, University of Wuerzburg, Emil-Fischer-Strasse 30, 97074 Wuerzburg, Germany}%
       and Tom Fischer\footrecall{bla} \\
			University of Wuerzburg}
	\end{center}

\bigskip

\begin{abstract}
Based on the work of \citet{Suz02}, we consider a generalization of Merton's asset valuation approach \citep{Mer74} in which two firms are linked by cross-ownership of equity and liabilities. 
 Suzuki's results then provide  no arbitrage prices of  firm values, which are derivatives of exogenous asset values. In contrast to the Merton model, the assumption of lognormally distributed assets does not result in lognormally distributed firm values, which also affects the corresponding probabilities of default. In a simulation study we see that, depending on the type of cross-ownership,  the lognormal model can lead to both, over- and underestimation of the actual probability of default of a firm under cross-ownership. In the limit, i.e. if the levels of cross-ownership tend to their maximum possible value,  these  findings can be shown theoretically as well. Furthermore, we  consider the default  probability of a firm in general, i.e. without a distributional assumption, and show that the lognormal model is often able to yield only a limited range of probabilities of default, while the actual probabilities may take any value between 0 and 1.
\end{abstract}


\section{Introduction}

Published in 1974, Merton's model of asset valuation  revolutionized     academic finance    as well as  the practice of both asset valuation and credit risk management. Since then, many  refinements and  extensions have been made (an overview may be found in a paper by \citet{Boh00}), but the crucial insight that the value of a firm's equity can be regarded as a European call option on the firm's asset value with strike price equal to the firm's debt, if the firm's financial structure is sufficiently simple, is still inherent to all these subsequent versions.
 Merton's approach not only provides an intuitive and tractable framework to value a firm's equity and debt, it also laid  the foundation to a wide class of credit risk models. As these models are characterized by the consideration of the firm's financial structure in order to derive the  probabilities of default, they are commonly referred to as ``structural models''. However, in their basic form, such models are applicable to a single firm only, and hence unable to explain the fact that credit events of  firms do not occur independently of each other, as  becomes evident in the work of \citet{Luc95}. Within structural models, one of the first  approaches of taking this  finding into account was to consider correlated asset values, which is for example described by \citet{Zho01}.  According to \citet{Gie04}, however, the consideration of asset correlations  only explains what he calls ``cyclical default correlation'' originating from the fact that firms are  subject to common macroeconomic factors. In contrast to that, correlation caused by what he calls ``credit contagion'' is not captured. Credit contagion arises from ``direct ties between firms'', as an example he describes the situation where one firm has given a trade credit to the other firm.  Although \citet{Gie04} writes that ``[i]t is easy to imagine that [...] economic distress of one firm can have an immediate adverse effect on the financial health of that firm's business partners'', this connection enters his model only indirectly by an incomplete information approach. Also \citet{Luc95} acknowledges the possibility that ``default correlation is caused if one firm is a creditor of another''. 
Although there are numerous articles on counterparty risk and financial contagion, see \citet{Jar95} and \citet{Duf96} and references therein,  current credit risk models  (an overview is provided by \citet{Cro00}) apparently do not take this fact into account explicitly, i.e. on a  structural level and in a multi-lateral way.  
One reason for this might be that only in the last decade models of asset valuation came to existence that directly include such relationships. \\

While Merton's model served as the basis for credit risk models for a single firm, the works of \citet{Eis01}, \citet{Suz02}, \citet{Els07} and \citet{Fis12} can be used to consider individual and joint probabilities of default by  direct incorporation of systemic risk caused by the structure of the firms' relations between each other. Financial claims and obligations   as described above can be subsumed under the general term of cross-ownership, which means that in a system of firms, these firms are linked to each other in that every firm's balance sheet contains financial assets or liabilities, no matter if short-term or long-term, issued by other firms in the system. 
In particular,  \citet{Eis01}, \citet{Suz02}, \citet{Els07} and \citet{Fis12} are all concerned with the problem of how to value such firms and any of their liabilities linked to each other by either  cross-ownership of equity and debt, or both. 
Under cross-ownership, credit contagion may not only occur unidirectional, but 
if a chain reaction forces an initially healthy firm to default, this event might revert to  the triggering firm, causing its financial situation to deteriorate even further, ``potentially a financial vicious circle'' \citep{Fis12}. In  \citet{Fis12}, a rather general setup of cross-ownership between  $n$ firms having $m$ liabilities is considered, which can include debt and derivatives, of differing seniority. The main result consists of an existence and uniqueness theorem of no-arbitrage prices of equity and liabilities, which can be computed by a fixed point iteration. The framework of \citet{Suz02} can be seen as a special case of \citet{Fis12}, since he examines the situation of $n$ firms having a single, homogeneous class of zero-coupon debt only. Thus, the work of \citet{Suz02} directly extends Merton's ideas to the case of two or more firms linked by cross-ownership. However, this approach has not yet been widely accepted or incorporated as a standard in  asset valuation or credit risk management.\\
 From both, an academic and  practical point of view, the question arises to what extent this neglect of cross-ownership between  firms  can affect the resulting firm values and estimated probabilities of default. Recall that \citet{Mer74} starts from  a single class of exogenous assets following a geometric Brownian motion, which means that asset values are lognormally distributed at maturity. Under cross-ownership,  however, the assets of a firm do not only consist of exogenous assets, but also of endogenous assets stemming from cross-ownership, for instance shares or bonds issued by another firm. It can be shown that firm values under cross-ownership, i.e. the total assets of a firm, are non-trivial  derivatives of exogenous asset values (see \citet{Suz02} and \citet{Fis12}, for example). Hence, firm values are generally not lognormally distributed anymore, in contrast to Merton's model. 
Returning  to the problem of determining probabilities of default, the assumption of generally lognormally distributed firm values when exogenous assets follow a lognormal distribution can  furthermore lead to incorrect probabilities of default, whether for the single firms or joint probabilities of default.  \\

Our work can be seen as a direct continuation of the work of \citet{Suz02} for two firms, since we also consider  cross-ownership scenarios with two firms having a single class of zero-coupon debt only. Based on Suzuki's formulas of equity and debt prices, we will first be  concerned with the resulting firm values and probabilities of default. 
More precisely, we  examine the  consequences of applying Merton's model of firm valuation to each firm separately, i.e. without consideration of cross-ownership.  Unfortunately, it seems to be impossible to obtain a closed-form solution of the distribution of firm values (and hence exact probabilities of default) under cross-ownership. Thus, we conduct a simulation study (cf. Section \ref{simstudy}) that compares the  probabilities of default resulting from Suzuki's model and from the lognormal distribution, the distribution of firm values resulting from Merton's approach.  In Section \ref{limit}, we provide a theoretical analysis of these probabilities in the limit, which means that we let the degree of cross-ownership converge to its maximum value. In this case, the distribution of firm values can be derived analytically, and the mathematical results match our empirical findings. In Section \ref{genprob}, we abandon any distributional assumptions with respect to exogenous asset values and analyze the probabilities of default under a rather general setup. Finally, Section \ref{outlook} summarizes our results and mentions some possible extensions.

\section{Firm Valuation with and without Cross-Ownership}
\subsection{Merton's model} \label{Merton}

In Merton's asset valuation model \citep{Mer74}, a single firm is assumed to have one class of exogenously priced assets $a$ and a certain amount of zero-coupon debt $d$ due at some future time $T$.  In this context, ``exogenously'' means that the value is independent of the firm's capital structure. At maturity, debt has to be paid back, but if the asset value has fallen  below the face value of debt at this time, the firm is said to be in default and all assets are handed over to the creditor. Thus, the creditor receives the minimum of $d$ and $a$, which we call the recovery value of debt, $r$. The value of equity, $s$, then is the value of the remaining assets, so, at maturity:  
\begin{eqnarray}
r &= \min\{d,a\} &= \text{recovery value of debt}, \label{rmert}\\
s &= (a-d)^+ \;\;&= \text{value of equity.} \label{smert}
\end{eqnarray}

\begin{table}
\begin{center}
\begin{tabular}{|c|c|}
\hline 
{Assets} &{ Liab.} \\ \hline 
$a$ & $s$ \\   & $r$ \\\hline
\end{tabular}
\caption{Single Firm: Balance sheet at maturity}
\label{bs-single}
\end{center}
\end{table}
The firm's balance sheet at maturity is given in Table \ref{bs-single}. 
Based on this balance sheet, we make the following definition.
\begin{definition} \label{fv-merton}
Based on \eqref{rmert} and \eqref{smert}, we define the \textit{firm value}~$v$ of a firm as the firm's total asset value:
\[ v:= r+s=a. 
\]
\end{definition}
A generalization of this firm value to the case of two firms linked by cross-ownership is derived in the next section.

\subsection{Suzuki's model}
\subsubsection{Cross-Ownership Fractions and Types of Cross-Ownership}
Let us now consider two firms linked by cross-ownership (``XOS''). Then  the assets of each firm do not only consist of an exogenous asset $a$, but also of financial assets issued by the other firm, for example in form of bonds or shares. As in the case of a single firm, we assume each firm to have a certain amount of zero-coupon debt with face value $d_1$ and $d_2$, respectively. Let $s_i$ and $r_i$ denote the no-arbitrage prices of equity and recovery value of debt of firm~$i$, $i=1,2$. 
Then the value of  a firm's assets originating from cross-ownership can be written as
\[
 \underbrace{M^{\rm{s}}_{ij} \cdot s_j}_{\text{cross-owned equity }} + \underbrace{M^{\rm{d}}_{ij} \cdot r_j,}_{\text{cross-owned debt}}
\]
where
$M^{\rm{s}}_{ij}$ and $M^{\rm{d}}_{ij}$ stand for the fraction that firm~$i$ owns of firm~$j$'s equity and debt, respectively. Note that the value of cross-owned debt is a fraction of the other firm's recovery value of debt, and not its face value of debt. \\

 In general, the so-called cross-ownership fractions $M^{\rm{s}}_{ij}$ and $M^{\rm{d}}_{ij}$ ($i=1,2; j=1,2; i\neq j$) can take values in the interval $[0,1]$. Based on their exact value, we define  three types of cross-ownership.
\begin{definition} \label{defxos}
The two firms are said to be linked by 
\begin{enumerate}
	\item  \label{typeI}\textbf{cross-ownership of equity only}, if
\[
M^{\rm{s}}_{1,2}>0,\; M^{\rm{s}}_{2,1}>0, \;M^{\rm{d}}_{1,2}=M^{\rm{d}}_{2,1}=0,
\]
that is  each firm holds a part of the other firm's equity, but none of its debt;
\item  \label{typeII}\textbf{cross-ownership of debt only}, if
\[
M^{\rm{s}}_{1,2}=M^{\rm{s}}_{2,1}=0,\; M^{\rm{d}}_{1,2}>0, \;M^{\rm{d}}_{2,1}>0,
\]
that is  each firm holds a part of the other firm's debt, but none of its equity;
\item  \label{typeIII} \textbf{simultaneous cross-ownership of equity and debt}, if
\[
\min\{M^{\rm{s}}_{1,2}, M^{\rm{s}}_{2,1},M^{\rm{d}}_{1,2}, M^{\rm{d}}_{2,1}\}>0,
\]
that is  each firm holds a part of both the other firm's equity and debt. 
\end{enumerate}
\end{definition}
\begin{remark}\label{extendeddef}
The definition of cross-ownership of both equity and debt (type \ref{typeIII}) could be extended to scenarios where exactly one of the four cross-ownership  fractions equals 0, or where firm~1 holds a part of firm~2's equity and firm~2 holds a part of firm~1's debt, or vice versa, i.e.
\begin{align*}
\min\{M^{\rm{s}}_{1,2},\;M^{\rm{d}}_{2,1}\}>0,&\quad  M^{\rm{s}}_{2,1}=M^{\rm{d}}_{1,2}=0,\\
\intertext{or}
M^{\rm{s}}_{1,2}=M^{\rm{d}}_{2,1}=0,&\quad   \min\{M^{\rm{s}}_{2,1},\;M^{\rm{d}}_{1,2}\}>0.
\end{align*}
In order to avoid case differentiations, we prefer Definition \ref{defxos}\ref{typeIII}.
\end{remark}

Note that  our definition of cross-ownership would not impose any restrictions with respect to the type of debt that is cross-owned. For example, a firm could  hold a derivative on any underlying considered in the model, e.g. exogenous assets. However, following  \citet{Suz02}, we will assume all liabilities to be zero-coupon-bonds with identical maturity. 

Furthermore, we will assume that no firm's equity or debt is completely owned by the other firm, but that some part of the equity and debt is held by a firm or investor outside of the system of the two firms. For the cross-ownership fractions, this implies that 
\[
\max\{M^{\rm{s}}_{1,2},M^{\rm{s}}_{2,1},M^{\rm{d}}_{1,2},M^{\rm{d}}_{2,1}\}<1.
\]
Furthermore, we suppose that no firm holds a part of its own equity or debt. \\
 
 The firms' balance sheets at maturity  under cross-ownership are given in Table~\ref{bs-xos}. 
\begin{table}%
\begin{center}
\begin{minipage}{0.4\textwidth}
\begin{flushright}
\begin{tabular}{|c|c|}
\multicolumn{2}{c}{Firm 1} \\ \hline
 \renewcommand{\arraystretch}{1.0}
{\footnotesize Assets} &{\footnotesize Liab.} \\ \hline  \renewcommand{\arraystretch}{1.5}
$a_1$ &  $s_1$\\ 
$\,M^{\rm{s}}_{1,2} \times s_2$ & $r_1$\\
$\,M^{\rm{d}}_{1,2}\times r_2$ & \\ \hline
\end{tabular}
\end{flushright}
\end{minipage} \hspace{2mm}
\begin{minipage}{0.4\textwidth}
\begin{flushleft}
\begin{tabular}{|c|c|}
\multicolumn{2}{c}{Firm 2} \\ \hline
 \renewcommand{\arraystretch}{1.0}
{\footnotesize Assets} &{\footnotesize Liab.} \\ \hline  \renewcommand{\arraystretch}{1.5}
$a_2$ & $s_2$ \\ 
$\,M^{\rm{s}}_{2,1} \times s_1$ & $r_2$\\
$\,M^{\rm{d}}_{2,1}\times r_1$ & \\ \hline
\end{tabular}
\end{flushleft}
\end{minipage}
\caption{Two firms under XOS: Balance sheets at maturity}
\label{bs-xos}
\end{center}
\end{table}
It is now clear that the value of firm~1 also depends on the financial health of firm~2: if firm~2 defaults, this will affect both, the value of its equity and the recovery value of its debt, which will possibly  be smaller than the actual outstanding amount. Therefore, the total  asset value of firm~1 will decrease and thus  firm~1 might also get into trouble, which again might affect firm~2 in a negative way. 
 If we applied Merton's model of firm valuation  to each firms separately in order to obtain no-arbitrage prices of equity and debt, we would ignore this circular dependence between the two firms. The work of \citet{Suz02} shows how to overcome this problem by applying Merton's idea to both firms simultaneously. 

\subsubsection{Suzuki's equations}

It is clear from Table \ref{bs-xos} that the total assets  of firm~$i$ consist of an exogenous and an endogenous part:
\begin{align*}
a_i^* &:= \underbrace{a_i}_{\parbox{1.5cm}{\centering exogenous assets}} +  \underbrace{M^{\rm{s}}_{ij}\,s_j + M^{\rm{d}}_{ij} r_j}_{\parbox{1.5cm}{\centering endogenous assets}} \geq 0, \quad \quad i=1,2;\; j=1,2;\; i\neq j,
\end{align*}
where ``endogenous'' means that the price is determined within the system of the two firms. \\
If we now apply Merton's approach to both firms simultaneously, we obtain in analogy to \eqref{rmert} and \eqref{smert} the following system of equations:
\begin{alignat}{2}
r_i &=  \min\{d_i,a_i^*\}& &=\min\{d_i,a_i + M^{\rm{s}}_{ij}\,s_j + M^{\rm{d}}_{ij} r_j\}, \label{eqri}\\
s_i &= (a_i^*-d_i)^+&&=(a_i+ M^{\rm{s}}_{ij}\,s_j + M^{\rm{d}}_{ij} r_j - d_i)^+, \label{eqsi}
\end{alignat}
with $i=1,2; \;j=1,2; \; i\neq j.$ \\
As before, the recovery value of debt of a firm still is the minimum of the firm's liability and total asset value, but under cross-ownership this recovery value now also depends on the other firm's equity value and recovery value of debt. Similarly,  the value of equity at maturity is now influenced by the other firm's equity and recovery value of debt. The recovery value of debt (as part of a solution of \eqref{eqri} and \eqref{eqsi}) and the total asset value $a_i^*$ of a firm are always non-negative. A proof can be found in \citet{Fis12}.\\ 
\citet{Suz02} solves this system of four non-linear equations with four unknowns. The resulting explicit formulas for $r_i$ and $s_i$ are given in the following lemma.

\begin{lemma} \label{defrs} The system \eqref{eqri}--\eqref{eqsi} is solved by
\begin{align*}
r_1 &= \begin{cases}
d_1, & (a_1,a_2) \in A_{\rm{ss}},\\
d_1, & (a_1,a_2) \in A_{\rm{sd}},\\
\frac{1}{1- M^{\rm{s}}_{1,2}M^{\rm{d}}_{2,1}} (a_1+ M^{\rm{s}}_{1,2} a_2 +(M^{\rm{d}}_{1,2}-M^{\rm{s}}_{1,2})d_2),\quad \quad\,\quad \quad \quad \quad \quad \quad \quad& (a_1,a_2) \in A_{\rm{ds}},\\
\frac{1}{1-M^{\rm{d}}_{1,2} M^{\rm{d}}_{2,1}} (a_1 + M^{\rm{d}}_{1,2} a_2),& (a_1,a_2) \in A_{\rm{dd}},
 \end{cases}\\
r_2 &= \begin{cases}
d_2, & (a_1,a_2) \in A_{\rm{ss}},\\
\frac{1}{1- M^{\rm{s}}_{2,1}M^{\rm{d}}_{1,2}} (M^{\rm{s}}_{2,1} a_1+  a_2 +(M^{\rm{d}}_{2,1}-M^{\rm{s}}_{2,1})d_1), \quad \quad\,\quad \quad \quad \quad \quad \quad \quad & (a_1,a_2) \in A_{\rm{sd}},\\
d_2,& (a_1,a_2) \in A_{\rm{ds}},\\
\frac{1}{1-M^{\rm{d}}_{1,2} M^{\rm{d}}_{2,1}} ( M^{\rm{d}}_{2,1}a_1 + a_2),& (a_1,a_2) \in A_{\rm{dd}},
 \end{cases}\\
s_1 &= \begin{cases}
\frac{1}{1-M^{\rm{s}}_{1,2} M^{\rm{s}}_{2,1}}(a_1+ M^{\rm{s}}_{1,2}a_2 + (M^{\rm{s}}_{1,2}M^{\rm{d}}_{2,1}-1)d_1+(M^{\rm{d}}_{1,2}-M^{\rm{s}}_{1,2})d_2, & (a_1,a_2) \in A_{\rm{ss}},\\
\frac{1}{1-M^{\rm{s}}_{2,1} M^{\rm{d}}_{1,2}}(a_1+M^{\rm{d}}_{1,2} a_2+(M^{\rm{d}}_{1,2} M^{\rm{d}}_{2,1}-1)d_1), &  (a_1,a_2) \in A_{\rm{sd}},\\
0, & (a_1,a_2) \in A_{\rm{ds}},\\
0,& (a_1,a_2) \in A_{\rm{dd}},
 \end{cases}\\
s_2 &= \begin{cases}
\frac{1}{1-M^{\rm{s}}_{1,2} M^{\rm{s}}_{2,1}}(M^{\rm{s}}_{1,2}a_1+ a_2 +(M^{\rm{d}}_{2,1}-M^{\rm{s}}_{2,1})d_1+ (M^{\rm{s}}_{2,1}M^{\rm{d}}_{1,2}-1)d_2, & (a_1,a_2) \in A_{\rm{ss}},\\
0, &  (a_1,a_2) \in A_{\rm{sd}},\\
\frac{1}{1-M^{\rm{s}}_{2,1} M^{\rm{d}}_{2,1}}(M^{\rm{d}}_{2,1}a_1+ a_2+(M^{\rm{d}}_{1,2} M^{\rm{d}}_{2,1}-1)d_2), & (a_1,a_2) \in A_{\rm{ds}},\\
0,& (a_1,a_2) \in A_{\rm{dd}},
 \end{cases}
\end{align*}
with (what we call ``Suzuki areas'')
\begin{align}
A_{\rm{ss}}= \{(a_1,a_2) \geq 0: a_1+M^{\rm{s}}_{1,2} a_2&\geq(1-M^{\rm{s}}_{1,2} M^{\rm{d}}_{2,1})d_1+(M^{\rm{s}}_{1,2} -M^{\rm{d}}_{1,2})d_2,\notag \\  M^{\rm{s}}_{2,1} a_1+a_2&\geq(M^{\rm{s}}_{2,1}-M^{\rm{d}}_{2,1})d_1+(1-M^{\rm{s}}_{2,1} M^{\rm{d}}_{1,2})d_2\}, \label{Ass}\\
A_{\rm{sd}} = \{(a_1,a_2) \geq 0: a_1+M^{\rm{d}}_{1,2} a_2 &\geq (1-M^{\rm{d}}_{1,2} M^{\rm{d}}_{2,1})d_1, \notag \\ M^{\rm{s}}_{2,1} a_1+a_2 &< (M^{\rm{s}}_{2,1}-M^{\rm{d}}_{2,1})d_1+(1-M^{\rm{s}}_{2,1} M^{\rm{d}}_{1,2} )d_2\}, \label{Asd}\\
A_{\rm{ds}} = \{(a_1,a_2) \geq 0: a_1+M^{\rm{s}}_{1,2} a_2 &<(1-M^{\rm{s}}_{1,2} M^{\rm{d}}_{2,1})d_1+(M^{\rm{s}}_{1,2} -M^{\rm{d}}_{1,2} )d_2, \notag \\ 
M^{\rm{d}}_{2,1} a_1+a_2 &\geq (1-M^{\rm{d}}_{1,2} M^{\rm{d}}_{2,1})d_2 \}, \label{Ads}\\
A_{\rm{dd}} = \{(a_1,a_2) \geq 0: a_1+M^{\rm{d}}_{1,2} a_2 &< (1-M^{\rm{d}}_{1,2} M^{\rm{d}}_{2,1})d_1, \notag \\ M^{\rm{d}}_{2,1} a_1+a_2 &< (1-M^{\rm{d}}_{1,2} M^{\rm{d}}_{2,1})d_2 \} \label{Add}.
\end{align} 
\end{lemma}

The exact derivation with proof may be found in \citet{Suz02}. Note that these formulas also hold for  the extended definition of cross-ownership of both equity and debt given in Remark \ref{extendeddef}.\\

Obviously, $r_i \leq d_i$ and $s_i\geq 0$. According to Lemma~\ref{defrs}, the functions $r_i(a_1,a_2)$ and $s_i(a_1,a_2)$ are section-wise defined, where the four sections on $\mathbb{R}_0^+ \times \mathbb{R}_0^+$ indicate if any of the two firms is in default or not: \\
By definition,  firm $i$ (i=1,2) is in default if its assets do not suffice to pay back all of its debt, i.e. if $a_i^* < d_i$. Equations \eqref{eqri} and \eqref{eqsi} imply $a_i^*=s_i+r_i$, and straightforward calculation yields 
\begin{align} \label{defdef}
\begin{aligned}
\text{firm 1 in default } &\Leftrightarrow a_1^* < d_1 \;\Leftrightarrow \;r_1 < d_1\; \Leftrightarrow\; (a_1,a_2) \in A_{\rm{ds}} \cup A_{\rm{dd}},\\
\text{firm 2 in default } &\Leftrightarrow a_2^* < d_2 \;\Leftrightarrow \;r_2 < d_2\; \Leftrightarrow\; (a_1,a_2) \in A_{\rm{sd}} \cup A_{\rm{dd}}.
\end{aligned}
\end{align}
This clarifies how to understand the notation $A_{c_1,c_2}$. 
If the exogenous asset value $(a_1,a_2)$ has fallen into a certain area $A_{c_1,c_2}$, firm~$i$'s condition is indicated by $c_i \in \{\rm{s}, \rm{d}\}$, where ``s'' stands for ``solvent'' and ``d'' for ``default''. 
 \\

\begin{remark}
Let $\mathbf{r}:=(r_1,r_2)^{\rm{T}}$, $\mathbf{s}:=(s_1,s_2)^{\rm{T}}$ and $\mathbf{d}=(d_1,d_2)^{\rm{T}}$. Then \eqref{eqri} and \eqref{eqsi} can be written as
\begin{align} 
\mathbf{r}&= \min \{\mathbf{d}, \mathbf{a}+ \mathbf{M}^{\rm{d}} \mathbf{r} + \mathbf{M}^{\rm{s}} \mathbf{s} \}, \label{mat.r}\\
\mathbf{s}&=(\mathbf{a}+ \mathbf{M}^{\rm{d}} \mathbf{r} +\mathbf{M}^{\rm{s}} \mathbf{s} -\mathbf{d})^+, \label{mat.s}
\end{align}
where $\mathbf{a}=(a_1,a_2)^{\rm{T}}$ and 
\[
\mathbf{M}^{\rm{s}}=\left( \begin{matrix}0 & M^{\rm{s}}_{1,2} \\ M^{\rm{s}}_{2,1} & 0  \end{matrix} \right), \quad \mathbf{M}^{\rm{d}}=\left( \begin{matrix}0 & M^{\rm{d}}_{1,2} \\ M^{\rm{d}}_{2,1} & 0  \end{matrix} \right),
\] and it follows that
\begin{equation} 
\mathbf{a} + \mathbf{M}^{\rm{d}} \mathbf{r} + \mathbf{M}^{\rm{s}} \mathbf{s} = \mathbf{r} + \mathbf{s}. \label{mat.balance}
\end{equation}
A generalization of \eqref{mat.r}, \eqref{mat.s} and \eqref{mat.balance} to the case of $n$ firms with $m$ outstanding liabilities each, can be found in \citet{Els07} and \citet{Fis12}. 
\end{remark}

\subsubsection{Firm value under Cross-Ownership}  \label{sim12}

It can be shown that $r_i$ and $s_i$ given in Lemma~\ref{defrs} are continuous functions of $a_1$ and $a_2$. Thus, the recovery value of debt and the value of equity under cross-ownership are derivatives of exogenous asset values:
\begin{equation} \label{derivatives}
r_i=r_i(a_1,a_2),\quad \quad s_i= s_i(a_1,a_2),\end{equation}
 just as in the Merton model.\\

In the Merton model for a single firm, we defined the firm value $v$ as the sum of equity value and recovery value of debt (cf. Definition~\ref{fv-merton}). 
Definition~\ref{fv-suz} transfers this definition to the case of cross-ownership.
\begin{definition}\label{fv-suz} 
Using the notation of Lemma~\ref{defrs}, the \textit{firm value} $v_i$ of  firm~$i$ under cross-ownership equals the firm's total asset value, i.e. 
\begin{align*}
v_i :&=  a_i + M^{\rm{s}}_{ij}\,s_j + M^{\rm{d}}_{ij} r_j. 
\end{align*} 
\end{definition}

From \eqref{eqri}, \eqref{eqsi} and  \eqref{derivatives} we obtain
\begin{align}
v_i=v_i(a_i,a_j) &=  a_i + M^{\rm{s}}_{ij}\,s_j(a_i,a_j) + M^{\rm{d}}_{ij} r_j(a_i,a_j) \notag\\
&= s_i(a_i,a_j)+r_i(a_i,a_j), \label{vsr}
\end{align}
which means that the firm value is also a derivative of exogenous asset values, and 
 Lemma~\ref{defrs}  yields
\begin{align} \label{defv1}
v_1 &= \begin{cases}
\frac{1}{1-M^{\rm{s}}_{1,2} M^{\rm{s}}_{2,1}}(a_1+ M^{\rm{s}}_{1,2} a_2 + M^{\rm{s}}_{1,2}(M^{\rm{d}}_{2,1}-M^{\rm{s}}_{2,1})d_1+(M^{\rm{d}}_{1,2}-M^{\rm{s}}_{1,2})d_2), & (a_1,a_2) \in A_{\rm{ss}},\\
\frac{1}{1-M^{\rm{s}}_{2,1} M^{\rm{d}}_{1,2}}(a_1+M^{\rm{d}}_{1,2} a_2+M^{\rm{d}}_{1,2}( M^{\rm{d}}_{2,1}-M^{\rm{s}}_{2,1})d_1), & (a_1,a_2) \in A_{\rm{sd}},\\
\frac{1}{1- M^{\rm{s}}_{1,2}M^{\rm{d}}_{2,1}} (a_1+ M^{\rm{s}}_{1,2} a_2 +(M^{\rm{d}}_{1,2}-M^{\rm{s}}_{1,2})d_2), & (a_1,a_2) \in A_{\rm{ds}},\\
\frac{1}{1-M^{\rm{d}}_{1,2} M^{\rm{d}}_{2,1}} (a_1 + M^{\rm{d}}_{1,2} a_2),& (a_1,a_2) \in A_{\rm{dd}},
 \end{cases}\\
v_2 &= \begin{cases} \label{defv2}
\frac{1}{1-M^{\rm{s}}_{1,2} M^{\rm{s}}_{2,1}}(M^{\rm{s}}_{2,1}a_1+  a_2+(M^{\rm{d}}_{2,1}-M^{\rm{s}}_{2,1})d_1 + M^{\rm{s}}_{2,1}(M^{\rm{d}}_{1,2}-M^{\rm{s}}_{1,2})d_2), & (a_1,a_2) \in A_{\rm{ss}},\\
\frac{1}{1-M^{\rm{s}}_{2,1} M^{\rm{d}}_{1,2}}(M^{\rm{s}}_{2,1}a_1+ a_2+( M^{\rm{d}}_{2,1}-M^{\rm{s}}_{2,1})d_1), & (a_1,a_2) \in A_{\rm{sd}},\\
\frac{1}{1- M^{\rm{s}}_{1,2}M^{\rm{d}}_{2,1}} (M^{\rm{d}}_{2,1}a_1+  a_2 +M^{\rm{d}}_{2,1}(M^{\rm{d}}_{1,2}-M^{\rm{s}}_{1,2})d_2), & (a_1,a_2) \in A_{\rm{ds}},\\
\frac{1}{1-M^{\rm{d}}_{1,2} M^{\rm{d}}_{2,1}} ( M^{\rm{d}}_{2,1} a_1 +a_2),& (a_1,a_2) \in A_{\rm{dd}}.
 \end{cases}
\end{align}

Furthermore, \eqref{defdef} and \eqref{vsr} imply
\begin{align}
\label{v1leqd1} 
\begin{aligned}
v_1 &< d_1 \Leftrightarrow (a_1,a_2) \in A_{\rm{ds}} \cup A_{\rm{dd}},\\
v_2 &< d_2 \Leftrightarrow (a_1,a_2) \in A_{\rm{sd}} \cup A_{\rm{dd}},
\end{aligned}
\end{align}
which yields a more intuitive representation of the Suzuki areas than the one given in \eqref{Ass} -- \eqref{Add}:
\begin{align} \label{gebiete}
\begin{aligned}
A_{\rm{ss}} = \{(a_1,a_2) \geq 0: v_1&\geq d_1, v_2 \geq d_2\}, \\
A_{\rm{sd}} = \{(a_1,a_2) \geq 0: v_1& \geq d_1, v_2< d_2\}, \\
A_{\rm{ds}} = \{(a_1,a_2) \geq 0: v_1&<d_1, v_2 \geq d_2\}, \\
A_{\rm{dd}} = \{(a_1,a_2) \geq 0: v_1&< d_1, v_2 < d_2\}. 
\end{aligned}
\end{align}

An example of the Suzuki areas is given in Figure \ref{areas}. Note that if $d_1\leq M^{\rm{d}}_{1,2}d_2$ and/or $d_2\leq M^{\rm{d}}_{2,1}d_1$, the areas $A_{\rm{ds}}$ and/or $A_{\rm{sd}}$ vanish.

\setlength{\unitlength}{1cm}
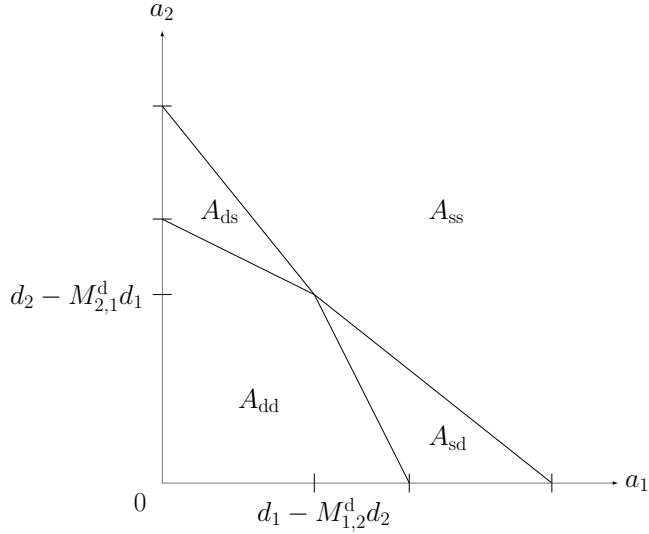
\begin{figure}%

\begin{center}
\scalebox{0.5}{
\begin{picture}(14,14)
\put(1,1){\vector(1,0){12}}
\put(1,1){\vector(0,1){12}}
\put(13.5,1){\makebox(0,0){\LARGE$a_1$}}
\put(1,13.5){\makebox(0,0){\LARGE$a_2$}}

\put(5,0.75){\line(0,1){0.5}}
\put(3.5,0){\text{\LARGE$d_1-M^{\rm{d}}_{1,2} d_2$}}
\put(-3,5.75){\text{\LARGE$d_2- M^{\rm{d}}_{2,1}d_1$}}
\put(0.75,6){\line(1,0){0.5}}
\put(0.25,0.25){\LARGE$0$}

\put(5,6){\line(2.5,-5){2.50}} 
\put(5,6){\line(1,-0.8){6.23}} 
\put(5,6){\line(-1,1.25){4}} 
\put(5,6){\line(-1,0.5){4}} 

\put(0.75,8){\line(1,0){0.5}}
\put(0.75,11){\line(1,0){0.5}}
\put(7.5,0.75){\line(0,1){0.5}}
\put(11.25,0.75){\line(0,1){0.5}}

\put(8,8){\LARGE$A_{\rm{ss}}$}
\put(3,3){\LARGE$A_{\rm{dd}}$}
\put(8,2){\LARGE$A_{\rm{sd}}$}
\put(2,8){\LARGE$A_{\rm{ds}}$}

\end{picture} }
\end{center}
\caption{Suzuki areas if $d_1> M^{\rm{d}}_{1,2} d_2$ and $d_2>M^{\rm{d}}_{2,1} d_1$}%
\label{areas}

\end{figure}

\bigskip
If the two firms have established cross-ownership of either equity or debt, the formula of $v_1$ given in \eqref{defv1} can be simplified. 
Under cross-ownership of equity only,  \eqref{defv1} reduces to 
\begin{align} \label{v1e}
v_1^{\rm{s}}:=\begin{cases}
\frac{1}{1-M^{\rm{s}}_{1,2} M^{\rm{s}}_{2,1} }(a_1+M^{\rm{s}}_{1,2} a_2-M^{\rm{s}}_{1,2} M^{\rm{s}}_{2,1} d_1-M^{\rm{s}}_{1,2} d_2), & (a_1,a_2) \in A_{\rm{ss}}, \\
a_1, & (a_1,a_2) \in A_{\rm{sd}} \cup A_{\rm{dd}},\\
a_1+M^{\rm{s}}_{1,2} a_2-M^{\rm{s}}_{1,2} d_2, & (a_1,a_2) \in A_{\rm{ds}}.
 \end{cases}
\end{align}

As we can see, $v_1^{\rm{s}}$ is a kind of section-wise defined linear combination of exogenous asset values $a_1$ and $a_2$ and face values of liabilities $d_1$ and $d_2$. Note that the coefficients of the liabilities are always non-positive. \\

Under cross-ownership of debt only, the value of firm~1 is given by
\begin{align} \label{v1d}
v_1^{\rm{d}}:=\begin{cases}
a_1+M^{\rm{d}}_{1,2} d_2, & (a_1,a_2) \in A_{\rm{ss}} \cup A_{\rm{ds}}, \\
a_1+M^{\rm{d}}_{1,2} a_2+M^{\rm{d}}_{1,2} M^{\rm{d}}_{2,1} d_1,& (a_1,a_2) \in A_{\rm{sd}},\\
\frac{1}{1-M^{\rm{d}}_{1,2} M^{\rm{d}}_{2,1} }(a_1+M^{\rm{d}}_{1,2} a_2), & (a_1,a_2) \in A_{\rm{dd}}.
 \end{cases}
\end{align}

Similar to $v_1^{\rm{s}}$, $v_1^{\rm{d}}$ is also a weighted sum of exogenous asset values, but in contrast to the case of cross-ownership of equity only in \eqref{v1e}, the face values of debt of both firms now contribute with  a non-negative sign.

By setting $M^{\rm{s}}_{1,2}=M^{\rm{s}}_{2,1}=0$ in \eqref{defv2}, one could also obtain formulas for $v_2^{\rm{s}}$ and $v_2^{\rm{d}}$, the value of firm~2 under cross-ownership of equity only and cross-ownership of debt only, respectively.

\subsection{Calculation of probabilities of default}

In the previous section, we saw that  the firm value is a function of the  exogenous asset values. In the following, we will assume these  exogenous asset values to be stochastic, which also turns the firm value $v$ into a random variable, because it is a continuous function of asset values. This is why we will denote asset values and firm values with  capital $A$s and $V$s, respectively, in the remainder.\\
We will assume exogenous assets to follow a bivariate geometric Brownian motion, similar to often extensions of the Merton model to the multivariate case. Thus, we have lognormally distributed exogenous asset values $A_1$, $A_2$ at maturity. We do not make any restrictions with respect to the correlation between $A_1$ and $A_2$. \\

Since the firm value equals the sum of exogenous and endogenous assets (cf. Definition \ref{fv-suz}), a firm is in default if and only if its firm value is smaller than the face value of its liabilities. Hence,
\begin{align} \label{pdefsuz}
P(\text{firm $i$ in default})=P(V_i < d_i).
\end{align}

Without cross-ownership, the assumption of lognormally distributed asset values would imply that firm values are also lognormally distributed because of $V_i=A_i$ in this situation (cf. Definition~\ref{fv-merton}), i.e. the evaluation of \eqref{pdefsuz} would be straightforward. 
But as we have seen in \eqref{defv1} and \eqref{defv2}, firm values are non-trivial derivatives of exogenous asset values under cross-ownership. Consequently,  the distribution of  firm values is a transformation of the lognormal distribution, which is generally not lognormal anymore.  However, we are not able to derive a closed-form solution of the resulting distribution, because, alongside other problems, there is no convolution theorem for lognormal distributions. \\

In this situation, one could ask to what extent the  probability of default of firm $i$ given in \eqref{pdefsuz} depends on whether the actual distribution of $V_i$ under cross-ownership  or the lognormal distribution is used. 
 Or expressed differently: what mistake (with respect to the resulting probabilities of default) do we make if we ignore that a part of the assets is  priced endogenously, and treat all assets as a single, homogeneous class of exogenous assets following a lognormal distribution which has the same first two moments as the actual firm value under cross-ownership? Since this approach would result in lognormally distributed firm values, this question essentially aims at the effects of applying Merton's model of firm valuation to both firms separately, despite the presence of cross-ownership.  

In the remainder, we will be concerned with  the comparison of probabilities of default obtained under both models.

\section{Simulation Study of Default Probabilities under Cross-Ownership} \label{simstudy}
\subsection{Setup and Parameter Values}  \label{parameter}

In order to get a first impression, we did a short simulation study for cross-ownership of equity only and cross-ownership of debt only (cf. Definition~\ref{defxos}) with the following parameters.

Exogenous assets of the two firms are independent and lognormally distributed at maturity $T=1$:
\begin{equation}
(A_1,A_2) \thicksim \mathcal{LN}(\boldsymbol{\mu}, \boldsymbol{\Sigma}) \label{logNVsim}
\end{equation}
with $\boldsymbol{\mu}=(\mu,\mu)^{\rm{T}}=(-0.5\sigma^2 + \ln(a),-0.5\sigma^2+\ln(a))$, $a>0$, and $\boldsymbol{\Sigma}= \bigl( \begin{smallmatrix}
\sigma^2 & 0 \\ 0 & \sigma^2
\end{smallmatrix} \bigr).$
This implies
\begin{align*}
 E(A_i)&= \exp(-0.5\sigma^2+\ln(a)+0.5\sigma^2)=a ,\\
\rm{Var}(A_i)&=\exp(-\sigma^2+2\ln(a)+\sigma^2) (\exp(\sigma^2)-1)=a^2 (\exp(\sigma^2)-1), \quad i=1,2.
\end{align*}
The coefficient of variation\footnote{For a random variable $X$ with mean $\mu$ and standard deviation $\sigma$, the coefficient of variation is defined as $\frac{\sigma}{\mu}$.} of $A_i$ ($i=1,2$) is given through $\sqrt{\exp(\sigma^2)-1}$.

Furthermore, the  liabilities of the two firms  have identical face values $d_1=d_2=:d$. Because of this kind of symmetry between the two firms, the main part of our study only analyzes probabilities of default of firm~1. 
Note that  any two setups for which the ratio $d/a$ is identical can be interpreted as the same setup under a different currency at a constant exchange rate. Thus, only the relative size of $d$ to $a$ is important, but not their absolute sizes. This is why we set $a=1$ in all our simulations and let only $d$ take different values. In particular, we have $E(A_i)=1$ and $\rm{Var}(A_i)=\exp(\sigma^2)-1$, $i=1,2$.

	In our  simulation study of default probabilities, we considered all possible combinations of $(M^{\rm{s}}_{1,2},M^{\rm{s}}_{2,1})$ with $M^{\rm{s}}_{1,2}, M^{\rm{s}}_{2,1}\in \{0.1, 0.2,\ldots,0.9\}$. Likewise for $V_1^{\rm{d}}$ and $(M^{\rm{d}}_{1,2},M^{\rm{d}}_{2,1})$. 
	
	The value of the liabilities, $d$, ran through \{0.1, 0.2, \ldots, 2.9, 3\}, which means that  
	\[
	\frac{\text{debt}}{\text{expected ex. assets}} = \frac{d}{a} \in \{0.1, 0.2, \ldots, 2.9, 3\}.
	\]	
The variance of logarithmized exogenous assets	$\sigma^2$ (cf. \eqref{logNVsim}) took values in $\{$0.00995, 0.22314, 0.44629, 0.69315, 1, 1.17865, 1.60944, 1.98100, 2.30259, 3.25810, 4.04743, 4.61512$\}$, which approximately resulted in  coefficients of variation of $A_i$ of $\{$0.1, 0.5, 0.75, 1, 1.31, 1.5, 2, 2.5, 3, 5, 7.5, 10$\}$.

	For every combination of parameters and both types of cross-ownership, 10,000 values of $(A_1,A_2)$ were simulated. 
	Based on \eqref{v1leqd1}, the probability of default under Suzuki's model was estimated by
\begin{align} \label{defps}
\hat{p}_{\text{S}} := \frac{\#\{(A_1,A_2) \in A_{\rm{ds}} \cup A_{\rm{dd}}\}}{10,000}.
\end{align}
The same simulated values of $(A_1,A_2)$  were used to calculate values for 	$V_1$, and from that an empirical distribution function $\hat{F}_{\rm{XOS}}$ of $V_1$. In order to determine the corresponding probability of default under the lognormal model,  we  approximated $\hat{F}_{\rm{XOS}}$  with a lognormal distribution. 
The parameters of this lognormal distribution were determined in analogy to the Fenton--Wilkinson method \citep{Fen60} of moment matching, which means that the first and second moments of $W$ were chosen such that they corresponded to the estimated first and second moments of $V_1$.  
	
	By \eqref{defps}, $\hat{p}_{\text{S}}$ was estimated with four decimal places only. 	
For a better comparison, we rounded the probabilities of default obtained from the lognormal model to four decimal places as well. These values will be denoted by $\hat{p}_{\text{L}}$.  As a measure for the discrepancy between the two models we used the relative risk RR of the two models, estimated by
\begin{align*}
\widehat{\text{RR}}:=\begin{cases}
\frac{\hat{p}_{\text{L}}}{\hat{p}_{\text{S}}}, & {\hat{p}_{\text{S}}}>0,\\
1, & {\hat{p}_{\text{S}}}=0 \text{ and } \hat{p}_{\text{L}}=0,\\
\infty,& {\hat{p}_{\text{S}}}=0 \text{ and } \hat{p}_{\text{L}}>0.
\end{cases}
\end{align*}
	The results of our simulation study are presented in the subsequent section.

\subsection{Results} \label{results}

First, we saw for both, cross-ownership of equity only and cross-ownership of debt only, that if the level of liabilities is chosen very small compared to $\sigma^2$, both models yield (rounded) estimated default probabilities  of 0, and the estimated relative risk ratios $\widehat{\text{RR}}$ equal 1. If $d/a$ is chosen very large compared to $\sigma^2$, we observe a similar effect, with the difference that now both models yield (rounded) estimated probabilities of default of 1.

However, note that the theoretical probabilities of default under either model  can never take a value of exactly 0 or exactly 1, since we assume exogenous assets to follow a (continuous) lognormal distribution. Hence, also if $d/a$ is chosen very large or small, the theoretical risk ratio is probably different from 1, but our short simulations cannot reveal whether we have to expect the lognormal model to over- or underestimate the actual risk in such scenarios. For very high levels of cross-ownership, the results of Section \ref{limit} will offer more insight.\\

When $d/a$ was chosen such that both (rounded) estimated probabilities of default were likely to lie in the open interval $(0,1)$, it seemed that with  increasing $\sigma^2$, the range of such values of $d/a$ became wider. In surface plots  we observed the following effects with respect to the cross-ownership fractions.

Under cross-ownership of equity only, we saw that, roughly speaking, the higher the cross-ownership fractions, the smaller the obtained values of $\widehat{\text{RR}}$. These values tended to be bigger than 1 or about 1 if $M^{\rm{s}}_{1,2}$ was small.  For $M^{\rm{s}}_{1,2}$ and $M^{\rm{s}}_{2,1}$  close to 1, we observed relative risks close to 0, i.e. the lognormal model then  underestimates the actual probability of default. An example is given in Figure \ref{prob1}(a), where the smallest value of $\widehat{RR}$ was 0.1779. For smaller values of $d/a$ we even obtained estimated relative risks of 0.

Under cross-ownership of debt, we observed opposite effects. Here, the  values of $\widehat{\text{RR}}$ were non-decreasing in the considered levels of cross-ownership, see Figure \ref{prob1}(b) for an example. For scenarios with high levels of cross-ownership (and $d/a$ chosen appropriately in the  sense explained earlier) we always obtained risk ratios greater than 1, i.e. the lognormal now overestimates the actual probability of default in these scenarios.

\begin{figure}%
\begin{center}
\subfigure[]{
\resizebox*{8cm}{!}{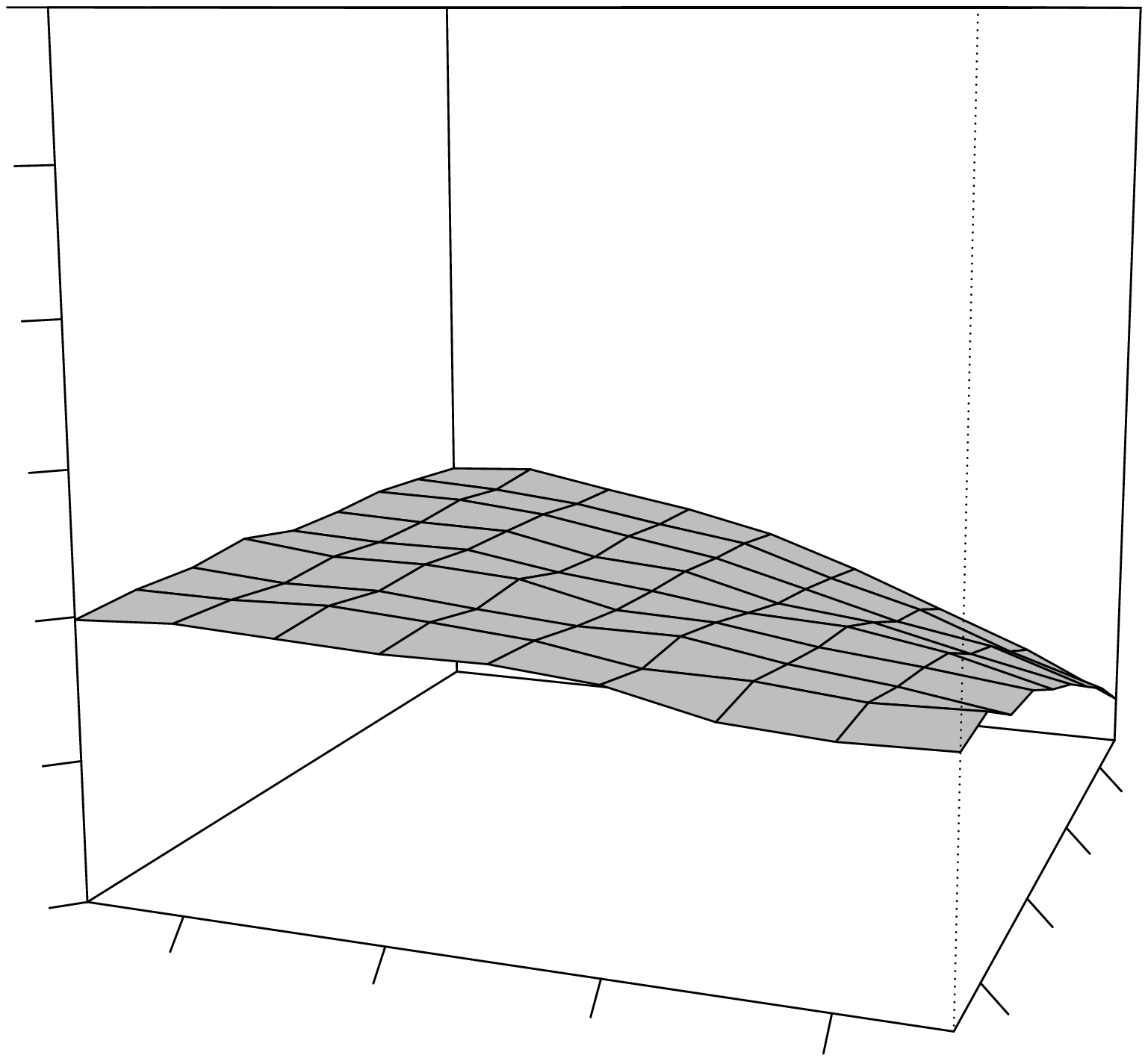}}%
\subfigure[]{
\resizebox*{8cm}{!}{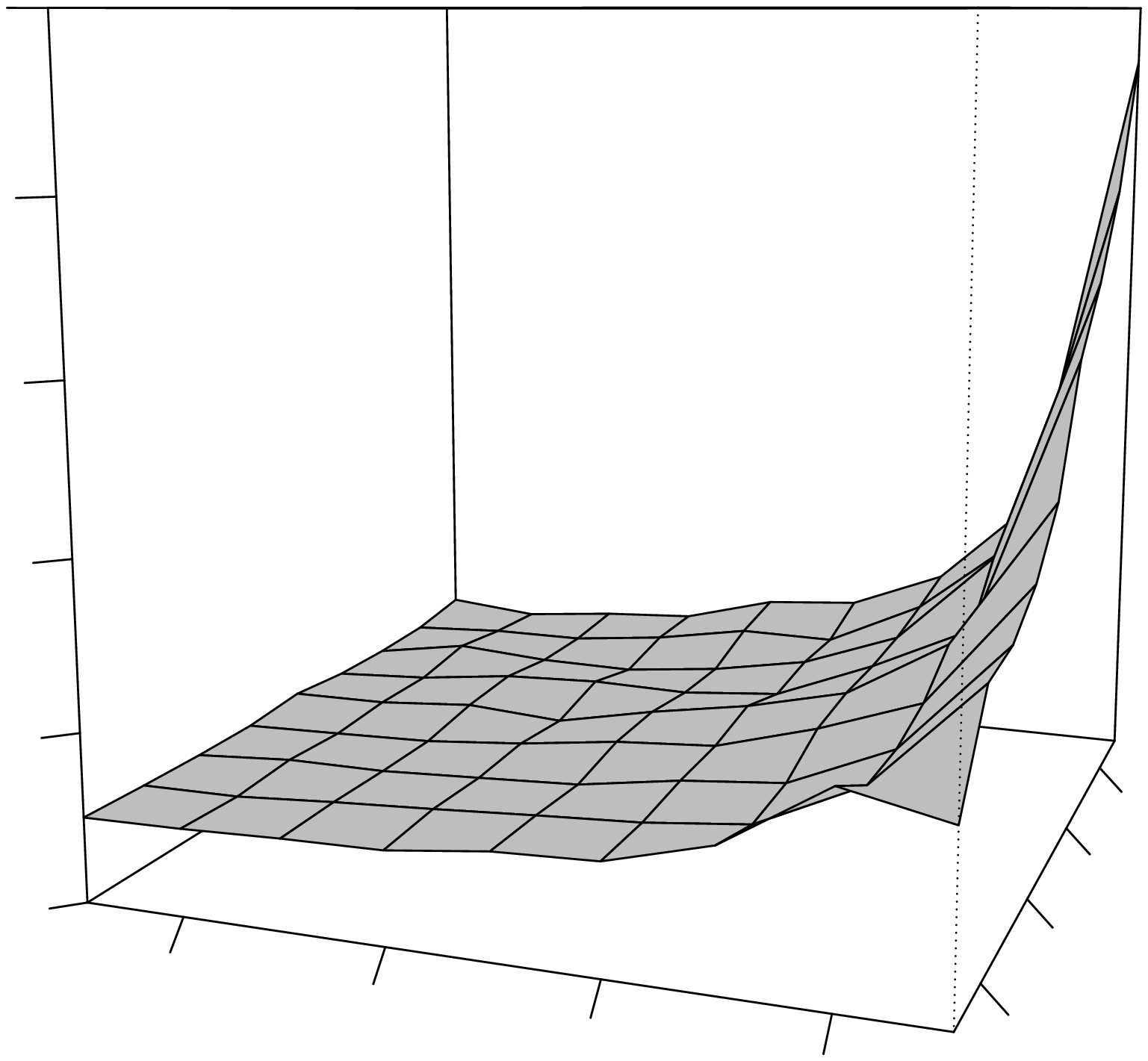}}%
\caption{Estimated rounded relative risk $\widehat{\text{RR}}$ in dependency of XOS-fractions  ($n=10,000$); (a) XOS of equity only, $\sigma^2=0.22314$, $d/a=0.7$; (b) XOS of debt only, $\sigma^2=1.60944$, $d/a$=0.4}
\label{prob1}
\end{center}
\end{figure}

\bigskip

As should have been expected, the difference between the two types of cross-ownership (in terms of the obtained risk ratios) especially becomes clear for scenarios with a high level of cross-ownership. Hence, we  fixed the cross-ownership fractions to 0.95 and had a closer look at the corresponding probabilities of default in a further short simulation study. Exogenous asset values were lognormally distributed with $a=1$ and $\sigma^2=1$ (cf. \eqref{logNVsim}), liabilities $d$ took values between 0.1 and 10 with steps of 0.1.  Every combination of parameters was repeated 100,000 times to obtain estimated probabilities of default (rounded to five decimal places now). The results for cross-ownership of equity only were such that the estimated relative risk was strictly smaller than 1 for any considered level of liabilities, whereas we always had estimated relative risks strictly greater than 1 under cross-ownership of debt only. 

Two examples can be found in Figure \ref{v1distr}. Note that the values of $d$ (0.9 and 1.6, resp.) were chosen such that the absolute difference between the estimated (rounded) probabilities of default was maximized. The estimated (rounded) probabilities under Suzuki's model and the lognormal model  are 0.51857 and 0.17464 in Figure \ref{v1distr}(a), and 0.02185 and 0.25530 in  Figure \ref{v1distr}(b). The corresponding  relative risks $\widehat{\text{RR}}$ amount to 0.33677 and 11.684.

\begin{figure}%
\begin{center}
\subfigure[]{
\resizebox*{8cm}{!}{\includegraphics[width=0.48\columnwidth]{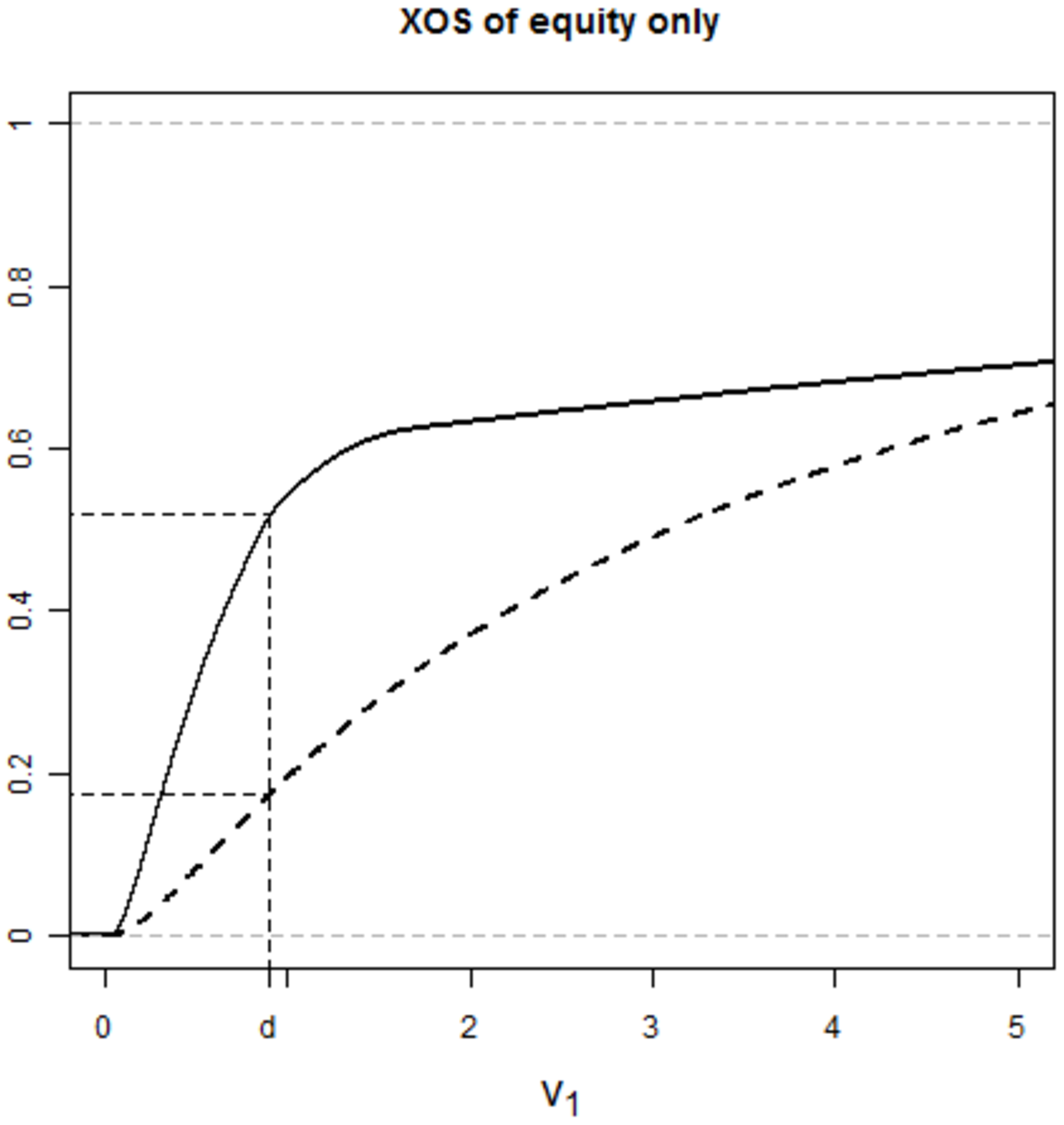}}}%
\subfigure[]{
\resizebox*{8cm}{!}{\includegraphics[width=0.48\columnwidth]{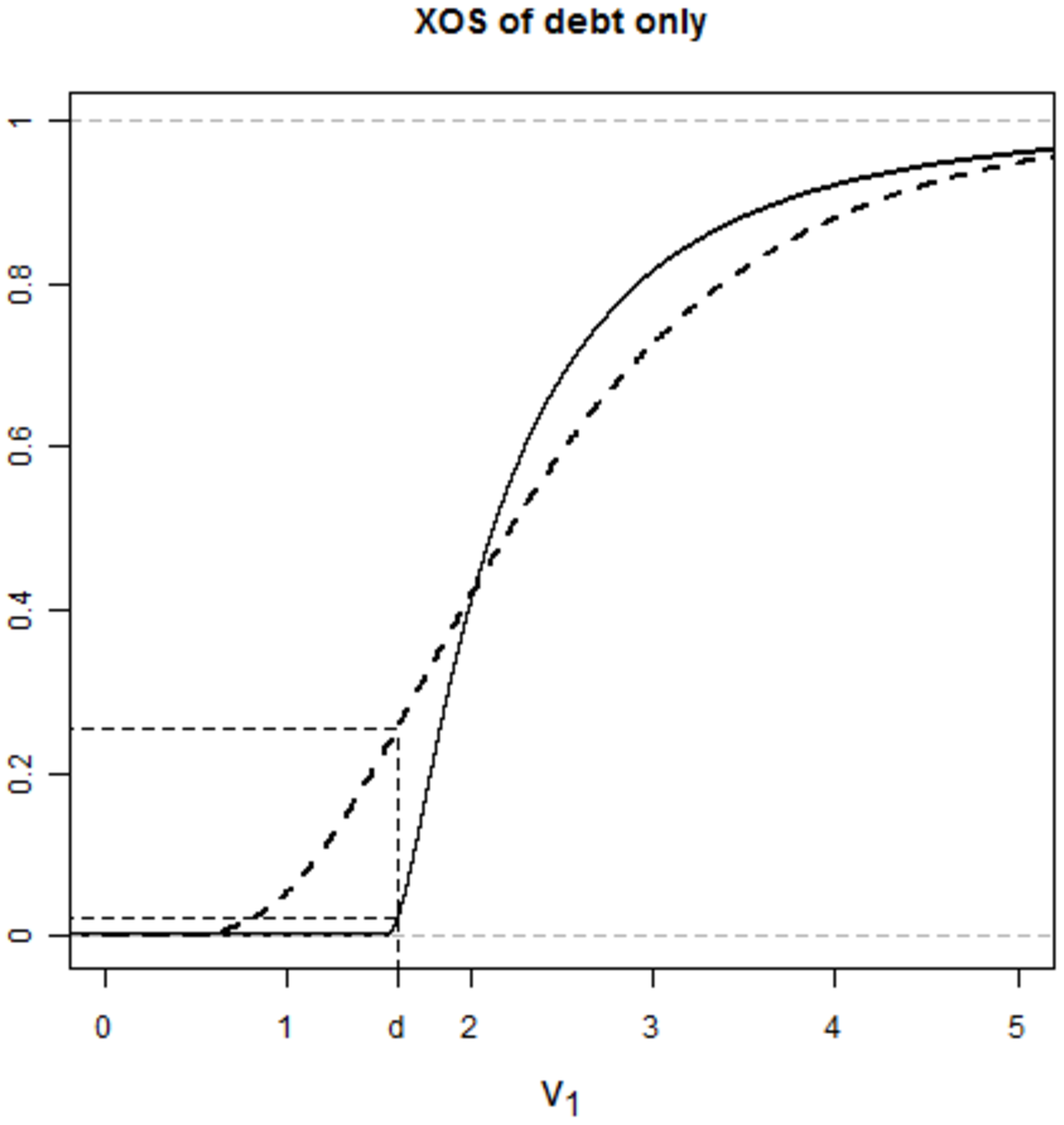}}}%
\caption{Probabilities of default, solid line: empirical distribution function of firm values $V_1$ resulting from Suzuki's model ($n=100,000$ iterations), dotted line: matched lognormal distribution; (a) XOS of equity only, $\sigma^2=1$, $d/a=0.9$; (b) XOS of debt only, $\sigma^2=1$, $d/a=1.6$.}
\label{v1distr}
\end{center}
\end{figure}
\bigskip

The insights gained in our simulations laid the foundation for the theoretical analyses in the subsequent sections.  In particular, we will be concerned with  probabilities of default of a firm under various assumptions.

\section{Limiting Probability of Default} \label{limit}
In our simulations we saw that if the two firms have established a high level of cross-ownership, the two types of cross-ownership seem to have opposite effects on the probabilities of default obtained under the lognormal model, compared to Suzuki's model. Unfortunately, we cannot compute the exact values under either model, because we can determine neither the distribution of $V_1$, nor its first and second moments in closed form. Hence, we cannot justify our findings theoretically. \\
However, the situation becomes tractable, if we let the cross-ownership fractions converge to 1. In this case, both the definition of the Suzuki areas given in \eqref{Ass}--\eqref{Add} and the formula of $V_1$ simplify, which makes an analytical approach possible. 
In this section, we will consider the ``limiting''  probability of default of firm~1 resulting from both, the Suzuki model and the corresponding matching lognormal model. This will be done separately for cross-ownership of equity only and of debt only. \\

As in our simulations, we assume that exogenous assets are lognormally distributed, i.e.
\begin{equation} \label{logNv}
(A_1,A_2) \thicksim \mathcal{LN}(\boldsymbol{\mu}, \boldsymbol{\Sigma})
\end{equation}
with $\boldsymbol{\mu}=(\mu_1,\mu_2)^{\rm{T}}$ and $\boldsymbol{\Sigma}= \bigl( \begin{smallmatrix}
\sigma_1^2 & \sigma_{12}  \\ \sigma_{12} & \sigma_2^2
\end{smallmatrix} \bigr).$
In particular, we  have $A_1, A_2>0$. Note that we do not impose any restrictions on $\boldsymbol{\mu}$ and $\boldsymbol{\Sigma}$.
Since we are only concerned with the default risk of firm~1, we set
\[
\mu:=\mu_1,\quad \sigma:=\sigma_1.
\]
In contrast to our simulations, we do not confine ourselves to the case of $d_1=d_2$. We only assume $d_1, d_2>0$ in order to exclude degenerate cases.

\subsection{XOS of equity only} \label{limitequity}

Let the firm value of firm~1 under cross-ownership of equity only, $V_1^{\rm{s}}$, be given by \eqref{v1e}.
If we consider the limit of $M^{\rm{s}}_{1,2}$ and $M^{\rm{s}}_{2,1}$ to 1, we are faced with the problem that for any given $(a_1,a_2) \in A_{\rm{ss}}$
\[
V_1^{\rm{s}} \big{|}_{A_{\rm{ss}}}= \frac{1}{1-M^{\rm{s}}_{1,2} M^{\rm{s}}_{2,1}}(a_1+ M^{\rm{s}}_{1,2} a_2 - M^{\rm{s}}_{1,2}M^{\rm{s}}_{2,1}d_1-M^{\rm{s}}_{1,2}d_2) \rightarrow \infty\quad\text{for } M^{\rm{s}}_{1,2},\,M^{\rm{s}}_{2,1} \rightarrow 1,
\]
since the limit of the term in brackets is always strictly positive, because on  $A_{\rm{ss}}$ (where, by definition, $M^{\rm{s}}_{1,2}$ and $M^{\rm{s}}_{2,1}$ are strictly smaller than 1) it holds that $a_1+a_2>d_1+d_2$ (cf. \eqref{Ass}).\\
Thus, if we want to evaluate the limiting probability of default under both models, this cannot be done by considering the pointwise limit of $V_1$ and the resulting limiting distribution.  
Instead, we will first calculate the probabilities of default under both models for $M^{\rm{s}}_{1,2}, M^{\rm{s}}_{2,1}<1$ and then consider the limits of these probabilities if cross-ownership fractions converge to 1.\\
Since firm~1 is in default if and only if its firm value is smaller than the face value of its debt at maturity,  we have under Suzuki's model by \eqref{gebiete}
\begin{align*}
P(V_1^{\rm{s}}<d_1)&=P(A_{\rm{ds}} \cup A_{\rm{dd}}) \\
&= P\big(\underbrace{\{(a_1,a_2)\geq 0: a_1 < d_1, a_2 <d_2+ \tfrac{1}{M^{\rm{s}}_{1,2}}(d_1-a_1)\}}_{=:A_{\rm{d.}}(M^{\rm{s}}_{1,2})}\big), 
\end{align*}
where the second equality follows from \eqref{Ads} and \eqref{Add} and $M^{\rm{d}}_{1,2}=M^{\rm{d}}_{2,1}=0$.     With $M^{\rm{s}}_{1,2}$ increasing, the set $A_{\rm{d.}}(M^{\rm{s}}_{1,2})$ becomes smaller, and it follows from the continuity of a probability measure that if $M^{\rm{s}}_{1,2}\rightarrow 1$,
\begin{equation} \label{probtrue}
 P(V_1^{\rm{s}} < d_1)\rightarrow P(\{(a_1,a_2)\geq 0:a_1 < d_1, a_1+a_2 \leq d_1+d_2\}) >0,
\end{equation}
where the strict positivity follows from the assumption that $d_1, d_2>0$.
 \\

 Let now
\begin{equation}
\tilde{V}_1^{\rm{s}} \thicksim \mathcal{LN}(\tilde{\mu}, \tilde{\sigma}^2),
\end{equation}
where $\tilde{\mu}$ and $\tilde{\sigma}^2$ are determined such that $E(\tilde{V}_1^{\rm{s}})=E(V_1^{\rm{s}})$ and $\rm{Var}(\tilde{V}_1^{\rm{s}})=\rm{Var}(V_1^{\rm{s}})$. Note that the square-integrability of $V_1^{\rm{s}}$ follows from the square-integrability of $A_1$ and $A_2$. This definition of $\tilde{V}_1^{\rm{s}}$ corresponds to the moment matching procedure applied in our simulations. \\
Under the lognormal model, the probability of default of firm~1 equals
\[
P(\tilde{V}_1^{\rm{s}} < d_1)=P(\tilde{V}_1^{\rm{s}} \leq d_1)=\Phi \left(\frac{\ln(d_1)-\tilde{\mu}}{\tilde{\sigma}} \right),
\]
where $\Phi$ stands for the standard normal distribution function.
If the cross-ownership fractions of equity converge to 1, this affects  both, expectation and variance of $V_1^{\rm{s}}$ and thus  also the parameters of $\tilde{V}_1^{\rm{s}}$, because
\begin{align}
\tilde{\mu} &=  \frac{1}{2} \ln \left( \frac{E(V_1^{\rm{s}})^4}{\rm{Var}(V_1^{\rm{s}})+E(V_1^{\rm{s}})^2} \right),\label{defmus}\\
\tilde{\sigma} &=  \ln \left(\frac{\rm{Var}(V_1^{\rm{s}})}{E(V_1^{\rm{s}})^2}+1 \right)^{0.5}. \label{defsigs}
\end{align}

More specifically, we have $\tilde{\mu} \rightarrow \infty$ for ${M^{\rm{s}}_{1,2},M^{\rm{s}}_{2,1} \rightarrow 1}$, and $\lim_{M^{\rm{s}}_{1,2},M^{\rm{s}}_{2,1} \rightarrow 1} \tilde{\sigma}<\infty$ by Lemma~\ref{limmusig} in the Appendix, i.e.
\[
\frac{\ln(d_1)-\tilde{\mu}}{\tilde{\sigma}} \rightarrow - \infty
\]
and thus 
\begin{equation} \label{problog}
P(\tilde{V}_1^{\rm{s}} < d_1)=\Phi \left(\frac{\ln(d_1)-\tilde{\mu}}{\tilde{\sigma}} \right) \rightarrow 0, \quad {M^{\rm{s}}_{1,2},\,M^{\rm{s}}_{2,1} \rightarrow 1}.
\end{equation}

Comparing \eqref{probtrue} and \eqref{problog}, we obtain the following proposition.
\begin{proposition} \label{finalequity}
Under cross-ownership of equity only with lognormally distributed exogenous asset values, the lognormal model underestimates the actual limiting default probability of a firm, i.e.
 \[
\lim_{M^{\rm{s}}_{1,2},M^{\rm{s}}_{2,1} \rightarrow 1} P(V_1^{\rm{s}} < d_1) > \lim_{M^{\rm{s}}_{1,2},M^{\rm{s}}_{2,1} \rightarrow 1} P(\tilde{V}_1^{\rm{s}} < d_1)=0.
\]
\end{proposition}

In our simulation study, this was already evident for cross-ownership fractions of 0.95 (cf. Section \ref{results}). \\

\begin{remark}
Note that all the results obtained in this section also hold without the assumption of lognormally distributed exogenous assets made in \eqref{logNv}, if we still approximate the distribution of $V_1^{\rm{s}}$ with a lognormal distribution. We only have to require the  distribution of exogenous assets to be continuous, non-negative, square-integrable  and to yield a strictly positive limiting probability of default $P(\{(a_1,a_2) \geq 0: a_1< d_1, a_1+a_2 \leq d_1+d_2\})$, 
   and both, a  strictly positive   expectation and variance of $A_1+A_2-d_1-d_2$ on $A_{\rm{ss}}^*$ (cf. Lemma \ref{existence.Ass} in the Appendix), i.e.  $E([A_1+A_2-d_1-d_2]\cdot 1_{\{A_1+A_2  \geq d_1+d_2\}})>0$ and $\rm{Var}([A_1+A_2-d_1-d_2]\cdot 1_{\{A_1+A_2 \geq d_1+d_2\}}) >0$
(cf. the proof of Lemma \ref{limmusig} in the Appendix).
\end{remark}

\subsection{XOS of debt only} \label{limitdebt}

Under cross-ownership of debt only,  firm values remain finite with probability 1 even if cross-ownership fractions converge to 1. Thus, we can  determine the limit  of $V_1^{\rm{d}}$ and compare the resulting probabilities of default under both models.  
Recall that we assume exogenous assets to follow a lognormal distribution given by \eqref{logNv}.\\

Based on \eqref{v1d} we can write
\begin{align}
&\begin{aligned}V_1^{\rm{d}} =  &1_{A_{\rm{ss}}}(A_1,A_2) \cdot \left(A_1+M^{\rm{d}}_{1,2} d_2 \right) \\
& \quad + 1_{A_{\rm{sd}}}(A_1,A_2) \cdot  \left(A_1+M^{\rm{d}}_{1,2} A_2+M^{\rm{d}}_{1,2} M^{\rm{d}}_{2,1} d_1 \right)\\\
& \quad +1_{A_{\rm{ds}}}(A_1,A_2) \cdot  \left(A_1+M^{\rm{d}}_{1,2} d_2  \right)\\
& \quad+1_{A_{\rm{dd}}}(A_1,A_2) \cdot  \left(\frac{1}{1-M^{\rm{d}}_{1,2} M^{\rm{d}}_{2,1} }(A_1+M^{\rm{d}}_{1,2} A_2) \right),
\end{aligned} \label{v1dind}
\end{align}
where $1_{A}$ stands for the indicator function of a set $A$. 
For the determination of the pointwise limit of $V_1^{\rm{d}}$ if $M^{\rm{d}}_{1,2},M^{\rm{d}}_{2,1}$ converge to 1, we first consider the limits of the indicator functions in \eqref{v1dind}. By Lemma \ref{existence} in the Appendix, their pointwise limits exist and we set
\begin{align}
\lim_{M^{\rm{d}}_{1,2}, M^{\rm{d}}_{2,1} \rightarrow  1} 1_{A_{ij}}&=:1_{A_{ij}^*},\quad ij \in \{\rm{ss},\rm{sd},\rm{ds},\rm{dd}\} \label{Aij*}
\end{align}
with $A_{\rm{dd}}^*= \{(0,0)\}$ by Lemma \ref{gebiete-empty} in the Appendix. Hence, $P(A_{\rm{dd}}^*)=0$ and 
\begin{align}
V_1^{\rm{d}^*}:=\lim_{M^{\rm{d}}_{1,2}, M^{\rm{d}}_{2,1} \rightarrow  1} V_1^{\rm{d}} &= 1_{A_{\rm{ss}}^*}(A_1,A_2) \cdot \left(A_1+ d_2 \right)\notag\\
& \quad + 1_{A_{\rm{sd}}^*}(A_1,A_2) \cdot  \left(A_1+ A_2+ d_1 \right) \label{v1d*}\\
& \quad +1_{A_{\rm{ds}}^*}(A_1,A_2) \cdot \left(A_1+ d_2  \right)  \hspace{2cm}P-a.s.\notag
\end{align}

Since almost sure convergence implies convergence in distribution, we have 
\begin{equation*}
\lim_{M^{\rm{d}}_{1,2},M^{\rm{d}}_{2,1} \rightarrow 1} P(V_1^{\rm{d}} < d_1)=P(V_1^{\rm{d}^*} < d_1).
\end{equation*}
 In order to determine the latter probability of default, we have to distinguish between the following three cases.

\subsubsection{$d_1 = d_2$} \label{d1=d2}
For $d_1 = d_2$ it follows from \eqref{v1d*} and Lemma \ref{gebiete-empty} in the Appendix that 
 $V_1^{\rm{d}^*}=A_1 + d_2$ $P-$a.s. 
Hence,  $V_1^{\rm{d}^*}$ follows a shifted lognormal distribution $\Lambda_{\mu, \sigma^2, \lambda}$ with shift parameter $\lambda=d_2$, which means that $\ln(V_1^{\rm{d}^*}-d_2) \thicksim \mathcal{N}(\mu, \sigma^2)$, and 
we obtain for the actual limiting probability of default that 
\[
P(V_1^{\rm{d}^*}< d_1)= P(V_1^{\rm{d}^*}< d_2)= 0.
\]
If we now match an unshifted, i.e. classical, lognormal distribution $\Lambda_{\tilde{\mu}, \tilde{\sigma}^2}$ to  $\Lambda_{\mu, \sigma^2, \lambda}$, it becomes clear that this distribution yields firm values lower or equal to $d_1$ with a strictly positive probability, because $d_1>0$.\\ Thus, if $d_1=d_2$ and if exogenous assets are lognormally distributed,  the lognormal model overestimates the actual risk of a firm under cross-ownership of debt if the cross-ownership fractions converge to 1.  \\
Recall that we had $d_1=d_2$ also in our simulations. In Section \ref{results} we saw that under cross-ownership of debt, the actual risk was overestimated already for cross-ownership fractions equal to 0.95.

\subsubsection{$d_1 < d_2$} \label{d1<d2}
If $d_1<d_2$,  it follows from \eqref{v1d*} and Lemma \ref{gebiete-empty} in the Appendix that
\begin{align*}
V_1^{\rm{d}^*}&=
(A_1 + d_2)\cdot  1_{A_{\rm{ss}}^*}+(A_1+A_2+d_1) \cdot 1_{A_{\rm{sd}}^*} \quad P-\text{a.s.},
\end{align*}
i.e. $V_1^{\rm{d}^*}>d_1$ with probability 1 (since $P(A_{\rm{dd}}^* \cup A_{\rm{ds}}^*)=0$) and thus $P(V_1^{\rm{d}^*}< d_1)=0$. As in the case of $d_1=d_2$, the lognormal model would yield a probability of default bigger than 0, so, under cross-ownership of debt only, the lognormal model again overestimates the actual risk of firm~1, if the corresponding cross-ownership fractions converge to 1.

\subsubsection{$d_1>d_2$} \label{d1>d2}
For this constellation, the situation is somewhat trickier.  
Equation \eqref{v1d*} and Lemma \ref{gebiete-empty} in the Appendix now yield
\begin{align*}
V_1^{\rm{d}^*}&= A_1+d_2\quad P-\text{a.s.},
\end{align*}
i.e. the firm value of firm~1 again follows a shifted lognormal distribution $\Lambda_{\mu,\sigma^2,\lambda}$ with shift parameter $\lambda=d_2$.
Interestingly, $V_1$ is independent of $d_1$, as long as this face value of debt is larger than $d_2$. 

Because of $d_1>d_2$, we now have 
\[
P(V_1^{\rm{d}^*} < d_1) > P(V_1^{\rm{d}^*} < d_2)=0,
\] i.e.
the limiting probability of default of firm~1  is strictly positive. So the argumentation used in the previous sections cannot be applied to this case. \\

 Let 
\begin{equation} \label{defv1d}
{\tilde{V}_1}^{\rm{d}^*} \thicksim \mathcal{LN}(\tilde{\mu}, \tilde{\sigma}^2),
\end{equation}
where $\tilde{\mu}$ and $\tilde{\sigma}^2$ are determined such that 
$E(\tilde{V}_1^{\rm{d}^*})=E(V_1^{\rm{d}^*})= E(A_1)+d_2$ and $\rm{Var}(\tilde{V}_1^{\rm{d}^*})=\rm{Var}(V_1^{\rm{d}^*})=\rm{Var}(A_1)$. \\

Straightforward calculations yield
\begin{align*}
\tilde{\mu}=&\frac{1}{2} \ln \left(\frac{(E(A_1)+d_2)^4}{\rm{Var}(A_1)+(E(A_1)+d_2)^2} \right) > \frac{1}{2} \ln \left( \frac{E(A_1)^4}{\rm{Var}(A_1)+E(A_1)^2} \right) =\mu, \\
\tilde{\sigma}^2=&\ln \left(\frac{\rm{Var}(A_1)}{(E(A_1)+d_2)^2}+1\right) < \ln\left(\frac{\rm{Var}(A_1)}{E(A_1)^2}+1 \right) =\sigma^2,
\end{align*}
where the last inequality follows from $d_2>0$.
 \\

Then we have the following limiting probabilities of default:
\begin{align*}
P(V_1^{\rm{d}^*}< d_1) &= P(A_1 < d_1-d_2)= \Phi \left(\frac{\ln(d_1-d_2)-\mu}{\sigma}\right),\\
P(\tilde{V}_1^{\rm{d}^*} < d_1) &= \Phi \left(\frac{\ln(d_1)- \tilde{\mu} }{\tilde{\sigma}} \right).
\end{align*}
Thus, in the limit, the lognormal model overestimates the actual risk if and only if
\begin{alignat}{2}
&&\frac{\ln(d_1-d_2)-\mu}{\sigma} &<  \frac{\ln(d_1)- \tilde{\mu} }{\tilde{\sigma}} \notag \\
&\Leftrightarrow\;\; &\tilde{\sigma} \ln(d_1-d_2)-\sigma \ln(d_1)&< \tilde{\sigma}\mu-\sigma \tilde{\mu} \notag\\
&\Leftrightarrow &\frac{(d_1-d_2)^{\tilde{\sigma}}}{(d_1)^{\sigma}} &< \exp(\tilde{\sigma}\mu-\sigma \tilde{\mu}). \label{LHS}
\end{alignat}

Straightforward calculations show that the LHS of  \eqref{LHS} as a function of $d_1$ has a maximum value of
\begin{equation}
 \frac{\left(\frac{\tilde{\sigma}}{\sigma-\tilde{\sigma}} d_2 \right)^{\tilde{\sigma}}}{\left( \frac{\sigma}{\sigma- \tilde{\sigma}} d_2\right)^{\sigma}}=:\text{LHS}_{\max} \label{lhsmax}\end{equation}
taken in $d_1=\frac{\sigma}{\sigma- \tilde{\sigma}}d_2=:d_{1,\max}> d_2$ because of $\sigma>\tilde{\sigma}$.
Furthermore, 
\[
\lim_{d_1 \searrow d_2} \frac{(d_1-d_2)^{\tilde{\sigma}}}{(d_1)^{\sigma}}=0, \quad \lim_{d_1 \rightarrow \infty} \frac{(d_1-d_2)^{\tilde{\sigma}}}{(d_1)^{\sigma}} = 0,
\]
which implies that the LHS of \eqref{LHS} is a bell-shaped, continuous function of $d_1$ with domain $(d_2,\infty)$ and maximum value LHS$_{\max}$ taken in $d_1=d_{1,\max}$.

 Note that the RHS of \eqref{LHS} is independent of $d_1$. It can be shown (cf. Lemma~\ref{inequality} in the Appendix) that 
\[
\text{LHS}_{\max} = \frac{\left(\frac{\tilde{\sigma}}{\sigma-\tilde{\sigma}} d_2 \right)^{\tilde{\sigma}}}{\left( \frac{\sigma}{\sigma- \tilde{\sigma}} d_2\right)^{\sigma}} > \exp(\tilde{\sigma}\mu-\sigma \tilde{\mu}),
\]
independently of the exact values of $\mu$, $\sigma$ and $d_2$, which means that the maximum value of the LHS of \eqref{LHS} as a function of $d_1$ is always greater than the constant  $\exp(\tilde{\sigma}\mu-\sigma \tilde{\mu})$. 

 Thus (cf. Figure \ref{skizze}),  there are two values $d_1^*$ and $d_1^{**}$, $d_1^{*}< d_{1,\max}< d_1^{**}$, such that
\begin{equation} \label{schnittpunkte}
\frac{(d_1^*-d_2)^{\tilde{\sigma}}}{(d_1^*)^{\sigma}} = \frac{(d_1^{**}-d_2)^{\tilde{\sigma}}}{(d_1^{**})^{\sigma}} =\exp(\tilde{\sigma}\mu-\sigma \tilde{\mu}),
\end{equation}
i.e.  \eqref{LHS} holds if and only if $d_1< d_1^*$ or $d_1 > d_1^{**}$. In these cases, 
the lognormal model overestimates the actual risk, if the cross-ownership fractions of debt converge to 1.

 \savedata{\mydata}[
 {
{1., 0.}, {1.1, 0.0787986}, {1.2, 0.126788}, {1.3, 0.155691}, {1.4, 
0.17248}, {1.5, 0.181444}, {1.6, 0.18529}, {1.7, 0.18577}, {1.8, 
0.184039}, {1.9, 0.180867}, {2., 0.176777}, {2.1, 0.172125}, {2.2, 
0.167157}, {2.3, 0.162041}, {2.4, 0.156892}, {2.5, 0.151789}, {2.6, 
0.146787}, {2.7, 0.141919}, {2.8, 0.137207}, {2.9, 0.132666}, {3., 
0.1283}, {3.1, 0.124112}, {3.2, 0.120101}, {3.3, 0.116263}, {3.4, 
0.112594}, {3.5, 0.109086}, {3.6, 0.105735}, {3.7, 0.102532}, {3.8, 
0.0994716}, {3.9, 0.0965465}, {4., 0.09375}, {4.1, 0.0910756}, {4.2, 
0.088517}, {4.3, 0.0860682}, {4.4, 0.0837235}, {4.5, 0.0814773}, 
{4.6, 0.0793246}, {4.7, 0.0772604}, {4.8, 0.0752802}, {4.9, 
0.0733794}, {5., 0.0715542}, {5.1, 0.0698005}, {5.2, 0.0681147}, 
{5.3, 0.0664934}, {5.4, 0.0649334}, {5.5, 0.0634316}, {5.6, 
0.0619852}, {5.7, 0.0605914}, {5.8, 0.0592477}, {5.9, 0.0579517}, 
{6., 0.0567012}}]

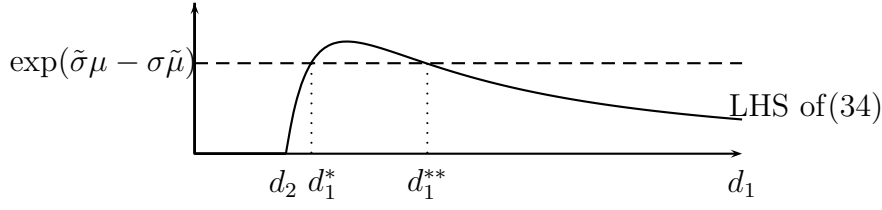
\begin{figure}
\begin{center}
 \psset{xunit=1.2cm,yunit=8cm,mathLabel=TRUE}
 \begin{pspicture}(-2,-0.1)(8,0.25)
 \psaxes[ticks=none,labels=none](1,0.25)
 \dataplot[plotstyle=curve,showpoints=FALSE]{\mydata}
 \psline{<->}(0,0.25)(0,0)(6,0)
\psline[linestyle=dashed](0,0.15)(6,0.15)
\rput(-1,0.15){$\exp(\tilde{\sigma}\mu-\sigma \tilde{\mu})$}
\rput(6,-0.05){$d_1$}
\rput(0.97,-0.05){$d_2$}
\rput(1.4,-0.05){$d_1^*$}
\rput(2.55,-0.05){$d_1^{**}$}
\psline[linestyle=dotted](1.28,0)(1.28,0.15)
\psline[linestyle=dotted](2.55,0)(2.55,0.15)
\rput(6.7,0.075){LHS of\eqref{LHS}}
 \end{pspicture}
\caption{Sketch of the LHS of \eqref{LHS} as a function of $d_1$.}
\label{skizze}
\end{center}
\end{figure}

\subsubsection{Conclusion for limiting risk under XOS of debt}
Our case differentiation with respect to the relative sizes of $d_1$ and $d_2$ 
can be summarized as follows.
\begin{proposition} \label{finaldebt}
Under cross-ownership of debt with lognormally distributed exogenous asset values, the lognormal model underestimates the actual limiting probability of default of a firm, i.e. 
 \[
\lim_{M^{\rm{d}}_{1,2},M^{\rm{d}}_{2,1} \rightarrow 1} P(V_1^{\rm{d}} < d_1) > P({\tilde{V}_1^{\rm{d}^*}} < d_1),
\]
if and only if
\[
d_1^*<d_1 <d_1^{**},
\]
with $d_1^*$ and $d_1^{**}$ given by \eqref{schnittpunkte}. In particular, the actual limiting default probability is overestimated if $d_1\leq d_2$.
\end{proposition}

Recall that under cross-ownership of equity, the lognormal model underestimated the actual limiting risk for every level of $d_1$ (cf. Proposition~\ref{finalequity}). So for $d_1 \leq d_2$, Proposition~\ref{finalequity} and Proposition~\ref{finaldebt} are both  confirmation  and extension to our empirical finding that the two types of cross-ownership have opposite effects on the probabilities of default.

\section{General Probabilities of Default} \label{genprob}

Having analyzed the probability of default of firm~1 if the respective cross-ownership fractions converge to 1 in the previous section, we will now examine the probability of default of a firm if cross-ownership fractions are strictly smaller than 1. 
For the case of cross-ownership of debt, the assumption of lognormally distributed exogenous assets proved to be crucial, whereas the results for the case of cross-ownership of equity also hold under far less restricting conditions. In the following, we will drop any distributional assumption with respect to exogenous asset values, we only require their distribution to be square-integrable and non-degenerate in a certain sense. This will be clarified later. In particular, we allow asset values to be zero. 
 Furthermore, our results will be valid for all three types of cross-ownership, we do not need a case differentiation as in Section \ref{limit}.\\

We set $A_{\rm{dd}} \cup A_{\rm{ds}} =: A_{\rm{d.}}$ and $A_{\rm{ss}} \cup A_{\rm{sd}} =: A_{\rm{s.}}$, i.e. $A_{\rm{d.}}$ and $A_{\rm{s.}}$ denote the regions where firm~1 is in default, or not.

Again, let $V_1$ be the (random) firm value of firm~1. 
According to the above partition of $\mathbb{R}^+_0 \times \mathbb{R}^+_0$, we also consider the distribution of $V_1$ as the weighted average of two conditional distributions on these areas, namely
\begin{equation}
P(V_1 \leq q) = P(V_1 \leq q \,|\, A_{\rm{d.}}) \times p +  P(V_1 \leq q \,|\, A_{\rm{s.}}) \times (1-p), \quad q \geq 0,\label{decomp}
\end{equation}
where 
\begin{equation} \label{defp}
p:= P((A_1,A_2) \in A_{\rm{d.}}).
\end{equation} 

In the following, we assume the conditional distributions of $V_1$ on $A_{\rm{d.}}$ and $A_{\rm{s.}}$ to be  fixed. Only their mixing parameter $p$ will vary. Of course, the ``total'' moments of $V_1$ on  $\mathbb{R}^+_0 \times \mathbb{R}^+_0$ then depend on $p$, which is indicated by the corresponding index:
\begin{align}
E_p(V_1) &= E(V_1 \,|\, A_{\rm{d.}})\times p + E(V_1 \,|\, A_{\rm{s.}})\times (1-p), \label{Ep}\\
	E_p(V_1^2)&= E(V_1^2 \,|\, A_{\rm{d.}})\times p + E(V_1^2 \,|\, A_{\rm{s.}})\times (1-p). \label{E2p}
\end{align}

Note that we do not make any specific assumptions with respect to the distribution of $(A_1,A_2): \Omega \rightarrow \mathbb{R}^+_0 \times \mathbb{R}^+_0$, we only require this distribution to imply 
\begin{equation} 0< \text{Var}_p(V_1) < \infty\quad \forall\; p \in [0,1]. \label{var.pos}
\end{equation}
A sufficient condition for \eqref{var.pos} to be met is that both conditional variances $\rm{Var}(V_1\,|\,A_{\rm{s.}})$ and $\rm{Var}(V_1\,|\,A_{\rm{d.}})$ are strictly positive\footnote{
This follows from
 \begin{align*}\rm{Var}_p(V_1)&=p \rm{Var}(V_1\,|\,A_{\rm{s.}})+(1-p)\rm{Var}(V_1\,|\,A_{\rm{d.}}) +p(1-p)[E(V_1\,|\,A_{\rm{s.}})-E(V_1\,|\,A_{\rm{d.}})]^2 \\ &\geq p \rm{Var}(V_1\,|\,A_{\rm{s.}})+(1-p)\rm{Var}(V_1\,|\,A_{\rm{d.}}).
\end{align*}}. \\
Because of $V_1 \geq 0$, \eqref{var.pos} also implies that $E_p(V_1)$ and $E_p(V_1^2)$ are strictly positive for all $p \in [0,1]$.
\\

Let us now consider the probabilities of default obtained from Suzuki's model and the lognormal model.

\subsection{Suzuki's model}
Under Suzuki's model we simply have
\begin{align*}
P(V_1 < d_1) = P((A_1,A_2) \in A_{\rm{d.}})= p
\end{align*}
by  \eqref{v1leqd1} and \eqref{defp}.

\subsection{Lognormal model}
For any $p \in [0,1]$, let $W_p$ be lognormally distributed with $E(W_p)=E_p(V_1)$ and $\text{Var}(W_p)=\text{Var}_p(V_1) = E_p(V_1^2)-E_p(V_1)^2$, i.e. 
\begin{align*}
W_p \thicksim \mathcal{LN}(\tilde{\mu}_p,\tilde{\sigma}_p^2) \notag
\end{align*}
with
\begin{align}
\tilde{\mu}_p &:= \ln\left(E(W_p)^2 \sqrt{\frac{1}{\text{Var}(W_p)+E(W_p)^2}} \,\right) = \ln\left(E_p(V_1)^2 \sqrt{\frac{1}{E_p(V_1^2)}}\, \right) \notag \\
&\,= \frac{1}{2} \ln \left(\frac{E_p(V_1)^4}{E_p(V_1^2)} \right),\\
\tilde{\sigma}_p^2&:= \ln \left(\frac{\text{Var}(W_p)}{E(W_p)^2}+1 \right) = \ln \left(\frac{E_p(V_1^2)}{E_p(V_1)^2} \right) >0. \label{sig2wp}
\end{align} 
Note that $\tilde{\sigma}_p^2$ is strictly positive for all $p \in [0,1]$ because of $E(W_p)>0$ and \eqref{var.pos}. \\

Under the lognormal model we then have 
\begin{align}
P(\text{firm 1 in default}) &= P(W_p < d_1) 
\notag \\
&= \Phi\left( \frac{\ln(d_1)-\tilde{\mu}_p }{\tilde{\sigma}_p}\right) \notag\\
&= \Phi\left( \frac{\ln(d_1)- \frac{1}{2} \ln \left(\frac{E_p(V_1)^4}{E_p(V_1^2)} \right)}{\ln \left(\frac{E_p(V_1^2)}{E_p(V_1)^2} \right)^{0.5}}\right). \label{phi}
\end{align} 
Setting
\begin{align}
E(V_1 \,|\, A_{\rm{d.}})-E(V_1 \,|\, A_{\rm{s.}})&=:x_1<0 \label{defx1}\\
E(V_1 \,|\, A_{\rm{s.}})&=:x_2 \geq d_1 \label{defx2}\\
E(V_1^2 \,|\, A_{\rm{d.}})-E(V_1^2 \,|\, A_{\rm{s.}})&=:y_1<0 \label{defy1}\\
E(V_1^2 \,|\, A_{\rm{s.}})&=:y_2 \geq d_1^2,\label{defy2}
\end{align}
it follows from \eqref{Ep}, \eqref{E2p} and \eqref{phi} that
\begin{align*}
P(W_p < d_1) 			&= \Phi\left( \frac{\ln(d_1)- \frac{1}{2} \ln \left(\frac{(p\times x_1+x_2)^4}{p \times y_1+y_2} \right)}{\ln \left(\frac{p \times y_1 +y_2}{(p \times x_1 +x_2)^2} \right)^{0.5}}\right).
			\end{align*}
			
Thus, the lognormal model underestimates the probability of default if and only if
\begin{align}
h(p):=\Phi\left( \frac{\ln(d_1)- \frac{1}{2} \ln \left(\frac{(p\times x_1+x_2)^4}{p \times y_1+y_2} \right)}{\ln \left(\frac{p \times y_1 +y_2}{(p \times x_1 +x_2)^2} \right)^{0.5}}\right) < p. \label{underest}
\end{align}
Since the denominator in \eqref{underest} equals $\tilde{\sigma}_p>0$  (cf. \eqref{sig2wp}), $h$ is always defined. Recall that $x_1$, $x_2$, $y_1$ and $y_2$ do not vary with $p$.

\subsection{Comparison}
\subsubsection{Values of $p$ close to or identical to 0 and 1} \label{general}
Let us consider \eqref{underest}. If $p=0$, which means that firm~1 is in default with probability 0 in Suzuki's model, we obtain for the lognormal model that $h(0)>0$, because the standard normal distribution function takes values in $(0,1)$ only. Since $h: [0,1] \rightarrow (0,1)$ is continuous in $p$, we know that there is a whole region $[0, \epsilon)$, $\epsilon>0$, with 
\begin{equation} \label{hp.groesser}
h(p) > p,\quad \quad p \in [0,\epsilon),
\end{equation}
where $\epsilon$ depends on $x_1,x_2,y_1$ and $y_2$. 
This can be interpreted as follows: If, for given $x_1,x_2,y_1$ and $y_2$, the actual probability of default for firm~1 is very small (i.e. smaller than $\epsilon(x_1,x_2,y_1,y_2)$), the lognormal model will overestimate this probability of default in this setup. 

Recall that in Sections \ref{d1=d2} and \ref{d1<d2}, we  observed a somewhat similar effect. There, cross-ownership fractions converged to 1, which resulted in an actual limiting default probability of 0, whereas the lognormal model  yielded a strictly positive limiting probability of default. For continuity reasons, there is a whole range of cross-ownership fractions such that the actual risk is overestimated.  
Hence, under cross-ownership of debt only (with $d_1 \leq d_2$), there are (at least) two ways of constructing scenarios leading to an overestimation of the actual default probability: first,  as done in  Sections \ref{d1=d2} and \ref{d1<d2}, we can alter the cross-ownership structure between the two firms such that the actual probability of default converges to 0, and second, we can transform the distribution of exogenous assets such that the actual probability of default converges to 0, as done in this section. In both approaches, the lognormal model yields a probability of default strictly greater than 0. 
 Note that the results of Sections  \ref{d1=d2} and \ref{d1<d2} hold without the assumption of exogenous assets following a lognormal distribution (cf. \eqref{logNv}).\\

Returning to \eqref{underest}, we obtain for $p=1$ that $h(1)<1$, i.e. there is an $\epsilon'(x_1,x_2,y_1,y_2)=:\epsilon' \in (0,1)$ such that
\begin{align} \label{hp.kleiner}
h(p) < p, \quad \quad p \in (\epsilon',1].
\end{align}
In this case, the lognormal model underestimates the probability of default. \\
By Proposition \ref{finalequity}, we see that under cross-ownership of equity only, underestimation of default probabilities can also be constructed by either a structural approach (i.e. letting cross-ownership fractions converge to 1) or a distributional approach (weighting the distribution of exogenous assets such that the actual probability of default converges to 1).

\begin{remark} \label{eps.mu}
The only assumption  we made about the distribution  of exogenous assets was that it implies $\text{Var}_p(V_1)>0$ for all  $p \in [0,1]$. Apart from this weak requirement, the above result is independent of the exact distribution of $(A_1,A_2)$ on $\mathbb{R}^+_0 \times \mathbb{R}^+_0$, in the sense that for any distribution $\mu$  on $\mathbb{R}^+_0 \times \mathbb{R}^+_0$ fulfilling \eqref{var.pos}, we can define a measure $P_{p,\mu}$ via
\begin{equation}
P_{p,\mu}(A):=p\, \frac{\mu(A \cap A_{\rm{d.}})}{\mu(A_{\rm{d.}})}+(1-p)\, \frac{\mu(A \cap A_{\rm{s.}})}{\mu(A_{\rm{s.}})}, \quad p \in [0,1],\; A \in \mathbb{R}^+_0 \times \mathbb{R}^+_0, \label{defppmu}
\end{equation}
and assume exogenous assets to be distributed according to $P_{p,\mu}$. 
If $p$ is chosen such that $p < \epsilon= \epsilon(\mu)$ or $p> \epsilon'=\epsilon'(\mu)$, \eqref{hp.groesser} or \eqref{hp.kleiner}, respectively, follow.\\
\end{remark}

Let us now consider $p \in (0,1)$.

\subsubsection{$d_1 \geq E_p(V_1)^2/E_p(V_1^2)^{0.5}$} \label{sec.d1.bigger}

For a given $p \in (0,1)$, let 
\begin{equation}
d_1 \geq \frac{E_p(V_1)^2}{E_p(V_1^2)^{0.5}}. \label{d1.bigger}
\end{equation}
 Because of Jensen's inequality we have 
\begin{equation} \label{condition}
\frac{ E_p(V_1)^2 }{E_p(V_1^2)^{0.5} }= E_p(V_1) \underbrace{\frac{ E_p(V_1) }{E_p(V_1^2)^{0.5} } }_{\leq 1} \leq E_p(V_1) \quad \forall \; p \in (0,1),
\end{equation}
so  \eqref{d1.bigger} is met if for example \begin{equation} E_p(V_1)\leq d_1. \label{suff}\end{equation} 
In Section \ref{feas}, we will see  how such a $V_1$ can be constructed. 
 \\

Under assumption \eqref{d1.bigger} we  have
\begin{equation*}
d_1^2 E_p(V_1^2) \geq E_p(V_1)^4,
\end{equation*}
which means that the numerator of \eqref{phi} is non-negative.  Thus,
\begin{equation*}
\Phi\left(\underbrace{ \frac{\ln(d_1)- \frac{1}{2} \ln \left(\frac{E_p(V_1)^4}{E_p(V_1^2)} \right)}{\ln \left(\frac{E_p(V_1^2)}{E_p(V_1)^2} \right)^{0.5}}}_{ \geq 0}\right) \geq 0.5,
\end{equation*}
i.e.  the lognormal model yields a probability of default of at least $0.5$ independently of the value of $p$,  as long as \eqref{d1.bigger} is met.\\

However, the initial assumption of  $d_1\geq E(V_1)^2/E(V_1^2)^{0.5}$ does not impose any restrictions on $p$, i.e. under the actual model, every probability of default can be obtained by choosing  suitable conditional distributions of $V_1$ on $A_{\rm{d.}}$ and $A_{\rm{s.}}$, respectively. This can be seen as follows.\\

Recall that
\[
E_p(V_1) = p \, \underbrace{E(V_1\,|\,A_{\rm{d.}})}_{< d_1}+(1-p)\, \underbrace{E(V_1\,|\,A_{\rm{s.}})}_{\geq d_1},
\]
which means that for any conditional distribution of $V_1$ on $A_{\rm{d.}}$, we only have to choose $P(V_1 \leq \cdot \,|\, A_{\rm{s.}})$ such that $E(V_1\,|\,A_{\rm{s.}})$ becomes small enough (i.e. close enough to $d_1$) to fulfill \eqref{suff}. This can always be achieved by putting enough mass on values close to $d_1$, which will be shown in Section \ref{feas}. Thus, the initial condition $d_1 \geq E_p(V_1)^2/E_p(V_1^2)^{0.5}$ can be met for any $p \in (0,1)$, if the distribution of $(A_1,A_2)$ on $\mathbb{R}^+_0 \times \mathbb{R}^+_0$ is chosen suitably. \\

Hence,  the probability of default $p$ in the Suzuki model  can be arbitrarily small, whereas the probability of default in the lognormal model is at least 0.5, assuming that $d_1\geq E_p(V_1)^2/E_p(V_1^2)^{0.5}$. In this case,  the actual risk is grossly overestimated.

However, if $p>0.5$, the actual risk might be underestimated.  \\

 \subsubsection{$d_1 \leq E_p(V_1)^2/E_p(V_1^2)^{0.5}$} \label{sec.d1.smaller}
Let now \begin{equation}
d_1 \leq \frac{E_p(V_1)^2}{E_p(V_1^2)^{0.5}} \label{dsmaller}
\end{equation}
for a given $p \in (0,1)$. Then we have 
\begin{equation*}
d_1^2 E_p(V_1^2) \leq E_p(V_1)^4,
\end{equation*}
which means that the numerator of \eqref{phi} is non-positive. Thus,
\begin{equation*}
\Phi\left(\underbrace{ \frac{\ln(d_1)- \frac{1}{2} \ln \left(\frac{E_p(V_1)^4}{E_p(V_1^2)} \right)}{\ln \left(\frac{E_p(V_1^2)}{E_p(V_1)^2} \right)^{0.5}}}_{ \leq 0}\right) \leq 0.5.
\end{equation*}

In contrast to that, the probability of default $p$ obtained from Suzuki's model can also take values larger than $0.5$. We show that for any $p \in (0,1)$ it is possible that \eqref{dsmaller} is fulfilled. By \eqref{condition}, a necessary condition for \eqref{dsmaller} is $E_p(V_1)>d_1$.

For some $E>d_1$, let the distribution of $(A_1,A_2)$ on $\mathbb{R}^+_0 \times \mathbb{R}^+_0$ be such that 
\begin{equation}
V_1= \begin{cases} 
	0.5 \,d_1, & \text{with probability } p,\\
	\frac{E-0.5\,p\,d_1}{1-p}, & \text{with probability } 1-p,
 \end{cases} \label{defv1-spec}
\end{equation}
that is there are only two firm values possible. In Section \ref{feas} we will see that there really is a distribution of $(A_1,A_2)$ on the positive quadrant that yields a distribution of $V_1$ as in \eqref{defv1-spec}.
\\Obviously,  $V_1\geq 0$ because of $E>d_1$, and $E_p(V_1)=E$, so the necessary condition for \eqref{dsmaller} is met.  Due to $\frac{E-0.5\,p\,d_1}{1-p}>d_1$, we indeed have $P(V_1 <d_1)=p$. Furthermore,
\begin{align*}
E_p(V_1^2)&= 0.25\,d_1^2\,p+\left(\frac{E-0.5\,p\,d_1}{1-p} \right)^2\,(1-p)\\
&=0.25\,d_1^2\,p + \frac{0.25\,d_1^2\,p^2-d_1\,p\,E+E^2}{1-p} = \frac{0.25\,d_1^2\,p-d_1\,p\,E+E^2}{1-p}
\intertext{and thus}
\frac{E_p(V_1)^2}{E_p(V_1^2)^{0.5}} \geq d_1 & \Leftrightarrow   \frac{E_p(V_1)^4}{E_p(V_1^2)} \geq d_1^2 \\
& \Leftrightarrow E^4 (1-p) \geq 0.25\,d_1^4 p-d_1^3 \,p\,E +E^2 d_1^2 \\
& \Leftrightarrow E^4 (1-p)  +d_1^3 \,p\,E -E^2 d_1^2 - 0.25\,d_1^4 p \geq 0.
\end{align*}
For given $p$, the inequality is always met if $E$ is chosen large enough.\\
For $p>0.5$ this means that the probability of default is underestimated if we use the lognormal model instead of Suzuki's model.

\subsubsection{Feasibility of required distributions of $V_1$} \label{feas}

In the previous sections, we saw that the lognormal model yields only a limited range of probabilities of default if the (conditional) distributions of $V_1$ are chosen suitably. Since the distribution of $V_1$ is a transformation of the distribution of exogenous assets, we have to make sure that it is in fact possible to choose the distribution of $(A_1,A_2)$  on $\mathbb{R}^+_0 \times \mathbb{R}^+_0$ such that
\begin{enumerate}
	\item $E(V_1\,|\,A_{\rm{s.}})$ is near to $d_1$ (Section \ref{sec.d1.bigger}). \label{nr1}
	\item for a given $p \in (0,1)$, the distribution of $V_1$ is of the form  \label{nr2}
	\[
	V_1= \begin{cases} 
	0.5\,d_1, & \text{with probability } p,\\
	\frac{E-0.5\,p\,d_1}{1-p}, & \text{with probability } 1-p,
 \end{cases} 
\]
	 (Section \ref{sec.d1.smaller}).
\end{enumerate}

Ad \ref{nr1}: 
If the conditional distribution of $(A_1,A_2)$ on $A_{\rm{s.}}$ is such that it has much mass near the ``border'' to $A_{\rm{d.}}$ (because we have $V_1=d_1$ on this border), the conditional expectation $E(V_1\,|\,A_{\rm{s.}})$  has the desired property in Section \ref{sec.d1.bigger}. \\

Ad \ref{nr2}:
Let $D:=\{(a_1,a_2):V_1(a_1,a_2)=0.5\, d_1\}=:D_{A_{\rm{dd}}} \cup D_{A_{\rm{ds}}}$ with 
\begin{align*}
D_{A_{\rm{dd}}}:=&\{(a_1,a_2) \in A_{\rm{dd}}: V_1(a_1,a_2)=0.5\, d_1\}\\
	=&\{(a_1,a_2) \in A_{\rm{dd}}: (a_1+ M^{\rm{d}}_{1,2}a_2)=0.5 (1-M^{\rm{d}}_{1,2} M^{\rm{d}}_{2,1}) d_1\},\\
D_{A_{\rm{ds}}}:=&\{(a_1,a_2) \in A_{\rm{ds}}: V_1(a_1,a_2)=0.5\, d_1\}\\
	=&\{(a_1,a_2) \in A_{\rm{ds}}: a_1+ M^{\rm{s}}_{1,2}a_2+(M^{\rm{d}}_{1,2}-M^{\rm{s}}_{1,2} )d_2 =0.5 (1-M^{\rm{s}}_{1,2} M^{\rm{d}}_{2,1}) d_1\}, 
\end{align*}
where we made use of \eqref{defv1}. \\
Since $E$ has to be chosen sufficiently large so that \eqref{dsmaller} is met, we can assume without loss of generality that $S:=\{(a_1,a_2): V_1(a_1,a_2) = \tfrac{E-0.5\,p\,d_1}{1-p}\} \subset A_{\rm{ss}}$. Then it follows  from \eqref{defv1} that
\begin{align*}
S=\{&(a_1,a_2) \in A_{\rm{ss}}:\\
\,&a_1+M^{\rm{s}}_{1,2}a_2+M^{\rm{s}}_{1,2}(M^{\rm{s}}_{2,1}-M^{\rm{d}}_{2,1}) d_1+(M^{\rm{s}}_{1,2}-M^{\rm{d}}_{1,2})d_2=(1-M^{\rm{s}}_{1,2} M^{\rm{s}}_{2,1} )\tfrac{E-0.5\,p\,d_1}{1-p}\},
\end{align*}
i.e. $S$ is a straight line in $A_{\rm{ss}}$. \\
In order to obtain the desired distribution of $V_1$, we thus only have to ensure that 
\begin{align*}
&P((A_1,A_2) \in  D) = p\\
&P((A_1,A_2) \in S) = 1-p\\
& P((A_1,A_2) \in \mathbb{R}^+_0 \times \mathbb{R}^+_0 \backslash (D \cup S)) =0,
\end{align*}
which can be constructed easily.

\subsection{Conclusion for the general case}
The above analysis shows that we cannot arrive at definite conclusions as to 
whether the lognormal model over- or underestimates the actual probability of default of a firm in the general case. Although \eqref{underest} provides an exact formula, we cannot solve this inequality for $p$ or the conditional moments of $V_1$ to gain further insight. \\ 
However, if $p=0$ or $p=1$, risk is systematically over- and underestimated, respectively. Further,  for given conditional distributions of $V_1$ on $A_{\rm{s.}}$ and $A_{\rm{d.}}$,  there is a whole interval $I_1:=[0,\epsilon)$ and hence a whole family  of distributions $P_p$ ($p \in I_1$) of $V_1$ (cf. \eqref{decomp} and Remark \ref{eps.mu}) such that the approximating lognormal model leads to an overestimation of the actual probability of default $p$. Similarly, there is an interval $I_2:=(\epsilon',1]$ with corresponding distributions of $V_1$  such that the approximating lognormal model leads to an underestimation of the actual probability of default $p \in I_2$.

If the expected firm value is smaller than the face value of debt, the lognormal model yields a probability of default of at least $0.5$, independently of the variance of the firm value. If the variance is small, the actual probability of default can be much smaller. On the other hand, there are also situations where the lognormal model grossly underestimates the actual risk of default.

\section{Summary and Outlook} \label{outlook}
For the case of two firms possessing a fraction of each other's of equity and/or debt, our analysis shows that Suzuki's method to evaluate these firms is much more appropriate than the application of Merton's model to each firm separately. Under cross-ownership, firm values are in general not lognormally distributed anymore, which is assumed under Merton's model. 
Our simulation study revealed  that the two models can yield relatively different probabilities of default (measured in terms of the relative risk ratio), if the two firms have established a high level of cross-ownership. The direction of the effect (i.e. over- or underestimation) depends on the considered type of cross-ownership. A theoretical analysis of our empirical findings confirmed that in the limit the lognormal model can lead to both over- and underestimation of the actual probability of default of a firm. 
Furthermore, we provide a formula that allows us to check for an arbitrary scenario of cross-ownership and any distribution of exogenous assets (on the positive quadrant) whether the approximating lognormal model will  over- or underestimate the related probability of default of a firm. In particular, any given distribution of exogenous asset values on the positive quadrant (non-degenerate in a certain sense)  can be transformed into  a new, ``extreme'' distribution of exogenous assets yielding such a high or low actual probability of default  that the approximating lognormal model will   under- and overestimate this risk, respectively.\\

Future research could aim at extending this analysis to the joint probability of default of the two firms.
Furthermore, one could consider the discrepancy between the univariate distribution functions of $V_1$ under Suzuki's model and the lognormal model in dependency of the model parameters, for example the realized type and level of cross-ownership and the face values of liabilities. For that, we are planning a further simulation study.
In a next step, it would be interesting to examine the bivariate distribution of $V_1$ and $V_2$ and the resulting dependency structure. In a short, at present unpublished analysis, we already gained a first impression which leads us to the conjecture that this dependency structure cannot be captured by the lognormal model, see Figure \ref{bivariate} for an example.

\begin{figure}%
\begin{center}
\subfigure[]{
\resizebox*{8cm}{!}{\includegraphics[width=0.48\columnwidth]{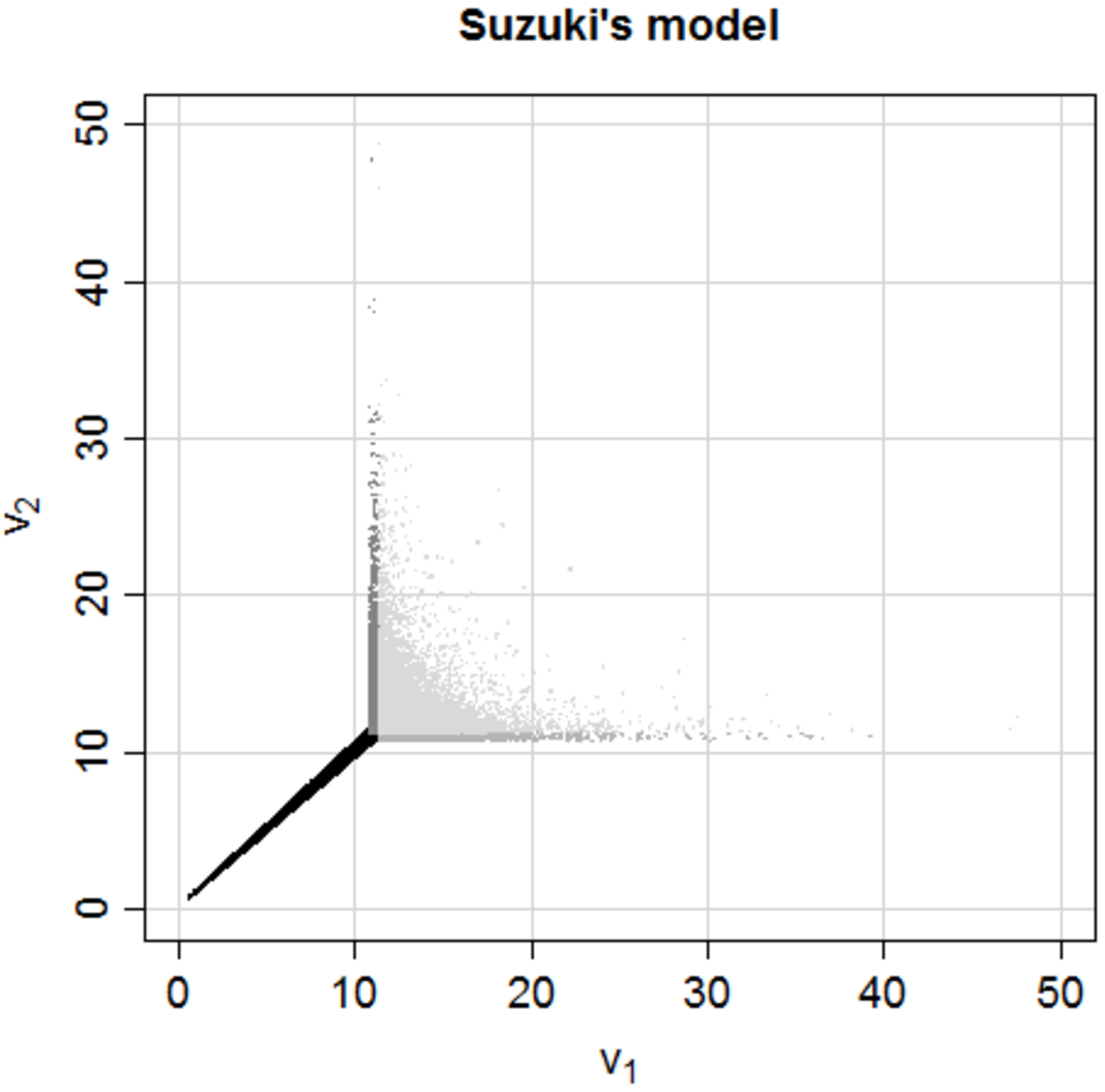}}}%
\subfigure[]{
\resizebox*{8cm}{!}{\includegraphics[width=0.48\columnwidth]{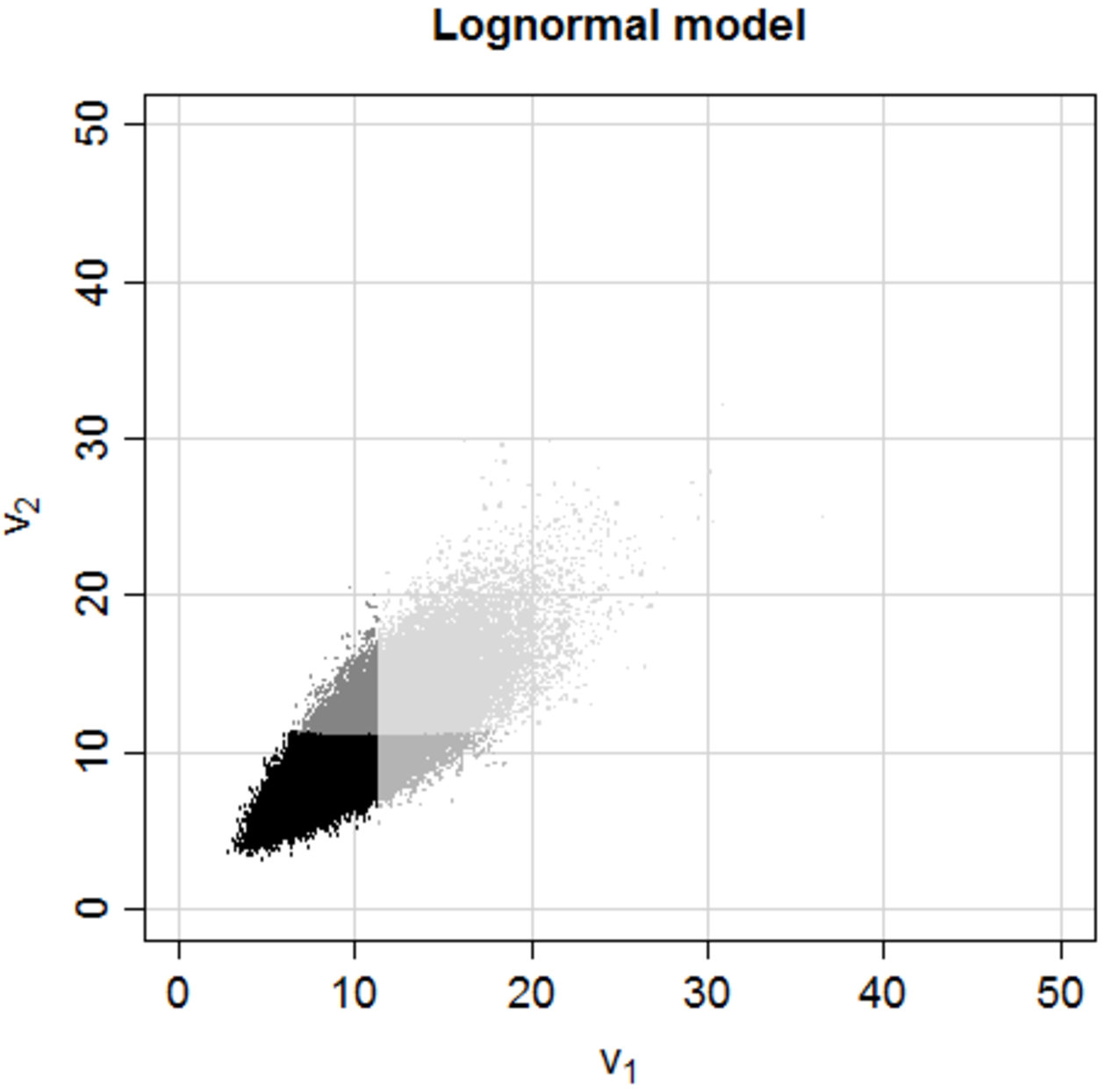}}}%
\caption{Scatterplot of bivariate firm values $(V_1,V_2)$ (XOS of debt only), stratified by Suzuki areas $A_{\rm{dd}}$ (black), $A_{\rm{ds}}$ (darkgrey),  $A_{\rm{sd}}$ (grey) and  $A_{\rm{ss}}$ (lightgrey);  $\sigma^2=1$, $d_1=11.3$, $M^{\rm{d}}_{1,2}=M^{\rm{d}}_{2,1}=0.95$, $n=100,000$; (a) firm values resulting from Suzuki's model  (b) firm values resulting from matched lognormal distribution.}
\label{bivariate}
\end{center}
\end{figure}

\pagebreak
\addcontentsline{toc}{section}{Bibliography}
\selectlanguage{english}
\bibliography{Bibliography_Paper_Theorie}

\begin{thebibliography}{13}
\providecommand{\natexlab}[1]{#1}
\providecommand{\url}[1]{\texttt{#1}}
\expandafter\ifx\csname urlstyle\endcsname\relax
  \providecommand{\doi}[1]{doi: #1}\else
  \providecommand{\doi}{doi: \begingroup \urlstyle{rm}\Url}\fi

\bibitem[Bohn(2000)]{Boh00}
J.~Bohn.
\newblock A survey of contingent-claims approaches to risky debt valuation.
\newblock \emph{The Journal of Risk Finance}, 1\penalty0 (3):\penalty0 53--70,
  2000.

\bibitem[Crouhy et~al.(2000)Crouhy, Galai, and Mark]{Cro00}
M.~Crouhy, D.~Galai, and R.~Mark.
\newblock A comparative analysis of current credit risk models.
\newblock \emph{Journal of Banking \& Finance}, 24\penalty0 (1--2):\penalty0
  59--117, 2000.

\bibitem[Duffie and Huang(1996)]{Duf96}
D.~Duffie and M.~Huang.
\newblock Swap rates and credit quality.
\newblock \emph{Journal of Finance}, 51\penalty0 (3):\penalty0 921--949, 1996.

\bibitem[Eisenberg and Noe(2001)]{Eis01}
L.~Eisenberg and T.~Noe.
\newblock Systemic risk in financial systems.
\newblock \emph{Management Science}, 47\penalty0 (2):\penalty0 236--249, 2001.

\bibitem[Elsinger(2007)]{Els07}
H.~Elsinger.
\newblock Financial networks, cross holdings, and limited liability.
\newblock working paper, Austrian National Bank, Economic Studies Division,
  2007.
\newblock URL \url{http://ssrn.com/abstract=916763}.

\bibitem[Fenton(1960)]{Fen60}
L.~Fenton.
\newblock The sum of log-normal probability distributions in scatter
  transmission systems.
\newblock \emph{IRE Transactions on Communication Systems}, 8\penalty0
  (1):\penalty0 57--67, 1960.

\bibitem[Fischer(2012)]{Fis12}
T.~Fischer.
\newblock No-arbitrage pricing under systemic risk: accounting for
  cross-ownership.
\newblock \emph{Mathematical Finance}, 2012.
\newblock DOI: 10.1111/j.1467-9965.2012.00526.x.

\bibitem[Giesecke(2004)]{Gie04}
K.~Giesecke.
\newblock Correlated default with incomplete information.
\newblock \emph{Journal of Banking \& Finance}, 28\penalty0 (7):\penalty0
  1521--1545, 2004.

\bibitem[Jarrow and Turnbull(1995)]{Jar95}
R.~Jarrow and S.~Turnbull.
\newblock Pricing derivatives on financial securities subject to credit risk.
\newblock \emph{Journal of Finance}, 50\penalty0 (1):\penalty0 53--85, 1995.

\bibitem[Lucas(1995)]{Luc95}
D.~Lucas.
\newblock Default correlation and credit analysis.
\newblock \emph{The Journal of Fixed Income}, 4\penalty0 (4):\penalty0 76--87,
  1995.

\bibitem[Merton(1974)]{Mer74}
R.~C. Merton.
\newblock On the pricing of corporate debt: the risk structure of interest
  rates.
\newblock \emph{Journal of Finance}, 29\penalty0 (2):\penalty0 449--470, 1974.

\bibitem[Suzuki(2002)]{Suz02}
T.~Suzuki.
\newblock Valuing corporate debt: the effect of cross-holdings of stock and
  debt.
\newblock \emph{Journal of the Operations Research Society of Japan},
  45\penalty0 (2):\penalty0 123--144, 2002.

\bibitem[Zhou(2001)]{Zho01}
C.~Zhou.
\newblock An analysis of default correlations and multiple defaults.
\newblock \emph{The Review of Financial Studies}, 14\penalty0 (2):\penalty0
  555--576, 2001.

\end{thebibliography}
\bibliographystyle{plainnat}

\appendix
\setcounter{secnumdepth}{0}
\section{Appendix: Some technical results}
 \setcounter{lemma}{0}
\renewcommand{\thelemma}{A\arabic{lemma}}

\begin{lemma} \label{limmusig}
Let $V_1^{\rm{s}}$ be given by  \eqref{v1e} in Section \ref{sim12}, let $(A_1,A_2)$ follow a lognormal distribution as given in \eqref{logNv} and let $\tilde{\mu}$ and $\tilde{\sigma}$ be defined as in \eqref{defmus} and \eqref{defsigs} in Section \ref{limitequity}, respectively. Then
\begin{align*}
\lim_{M^{\rm{s}}_{1,2},M^{\rm{s}}_{2,1} \rightarrow 1} \tilde{\sigma}  < \infty \quad\text{and} \quad
\tilde{\mu}   \rightarrow \infty \; \text{ for } {M^{\rm{s}}_{1,2},M^{\rm{s}}_{2,1} \rightarrow 1}.
\end{align*}
\end{lemma}
\begin{proof}
From 
\begin{equation} \label{v1ass}
V_1^{\rm{s}} \big{|}_{A_{\rm{ss}}}= \frac{1}{1- M^{\rm{s}}_{1,2} M^{\rm{s}}_{2,1}} (\underbrace{A_1+M^{\rm{s}}_{1,2} A_2- M^{\rm{s}}_{1,2}(M^{\rm{s}}_{2,1}d_1+d_2)}_{\geq 0 \text{ by } \eqref{Ass}}), \end{equation} 
we obtain
\begin{align}
E(V_1^{\rm{s}} \cdot 1_{A_{\rm{ss}}}) &=  \frac{1}{1- M^{\rm{s}}_{1,2}M^{\rm{s}}_{2,1} }E([A_1+M^{\rm{s}}_{1,2} A_2- M^{\rm{s}}_{1,2}(M^{\rm{s}}_{2,1}d_1+d_2)]\cdot 1_{A_{\rm{ss}}}), \label{ev1ass}\\
\rm{Var}(V_1^{\rm{s}}\cdot 1_{A_{\rm{ss}}}) &= \left(\frac{1}{1- M^{\rm{s}}_{1,2}M^{\rm{s}}_{2,1}}\right)^2\;\rm{Var}([A_1+M^{\rm{s}}_{1,2} A_2]\cdot 1_{{A_{\rm{ss}}}})\notag\\
&=\left(\frac{1}{1- M^{\rm{s}}_{1,2}M^{\rm{s}}_{2,1}}\right)^2 \left( E([A_1+M^{\rm{s}}_{1,2} A_2]^2\cdot  1_{A_{\rm{ss}}})-E([A_1+M^{\rm{s}}_{1,2} A_2]  \cdot 1_{A_{\rm{ss}}})^2\right). \label{varv1ass}
\end{align}
Let $1_{A_{\rm{ss}}^*}$ denote the limit of $1_{A_{\rm{ss}}}$ if $M^{\rm{s}}_{1,2},M^{\rm{s}}_{2,1} \rightarrow 1$. For its existence, see Lemma \ref{existence.Ass}. In particular, $1_{A_{\rm{ss}}^*} \geq 1_{A_{\rm{ss}}}$ for all $M^{\rm{s}}_{1,2},M^{\rm{s}}_{2,1} \in (0,1)$. 
Because of
\begin{align*}
[A_1+ M^{\rm{s}}_{1,2} A_2-M^{\rm{s}}_{1,2}(M^{\rm{s}}_{2,1}d_1+d_2)] \cdot 1_{A_{\rm{ss}}}& \leq [A_1+ M^{\rm{s}}_{1,2}A_2]\cdot 1_{A_{\rm{ss}}}\leq [A_1+A_2] \cdot 1_{A_{\rm{ss}}^*}, \\
[A_1+ M^{\rm{s}}_{1,2}A_2]^2\cdot 1_{A_{\rm{ss}}} &\leq  (A_1+A_2)^2\cdot 1_{A_{\rm{ss}}^*}
\quad \quad \text{for all } M_{1,2}^{\rm{s}}, M_{2,1}^{\rm{s}} \in (0,1),
\end{align*}
 the Dominated Convergence Theorem  implies that if $M^{\rm{s}}_{1,2},M^{\rm{s}}_{2,1} \rightarrow 1$,
\begin{align}
E([A_1+M^{\rm{s}}_{1,2} A_2- M^{\rm{s}}_{1,2}(M^{\rm{s}}_{2,1}d_1+d_2)] \cdot 1_{A_{\rm{ss}}})& \rightarrow E([A_1+A_2-d_1-d_2]\cdot 1_{A_{\rm{ss}}^*})<\infty, \label{limev1ass}\\
E([A_1+M^{\rm{s}}_{1,2} A_2] \cdot 1_{A_{\rm{ss}}}) &\rightarrow E([A_1+A_2] \cdot 1_{A_{\rm{ss}}^*})< \infty, \label{limvarv1ass1}\\
E([A_1+M^{\rm{s}}_{1,2} A_2]^2 \cdot 1_{A_{\rm{ss}}}) &\rightarrow E([A_1+A_2]^2 \cdot 1_{A_{\rm{ss}}^*})<\infty, \label{limvarv1ass}
\end{align}
i.e.  $\rm{Var}([A_1+M^{\rm{s}}_{1,2} A_2] \cdot 1_{A_{\rm{ss}}})\rightarrow \rm{Var}([A_1+ A_2]\cdot 1_{A_{\rm{ss}}^*})$ for  $M^{\rm{s}}_{1,2},M^{\rm{s}}_{2,1} \rightarrow 1$. Note that both, $E([A_1+A_2-d_1-d_2]\cdot 1_{A_{\rm{ss}}^*})$ and $\rm{Var}([A_1+ A_2]\cdot 1_{A_{\rm{ss}}^*})$, are strictly positive due to  the lognormal distribution of $(A_1,A_2)$ and   the fact that  $A_{\rm{ss}}^*=\{(a_1,a_2)\geq 0:a_1+a_2 \geq d_1+d_2\} \neq \emptyset$ (cf.  Lemma \ref{existence.Ass}). We obtain from  \eqref{ev1ass}--\eqref{limvarv1ass} that
\begin{align}
 E(V_1^{\rm{s}} \cdot 1_{A_{\rm{ss}}}), \rm{Var}(V_1^{\rm{s}} \cdot 1_{A_{\rm{ss}}}) &\rightarrow \infty,\quad  M^{\rm{s}}_{1,2},M^{\rm{s}}_{2,1} \rightarrow 1, \notag
\end{align}
and
\begin{align}
\frac{\rm{Var}(V_1^{\rm{s}} \cdot 1_{A_{\rm{ss}}})}{ E(V_1^{\rm{s}} \cdot 1_{A_{\rm{ss}}})} &= \frac{1}{1- M^{\rm{s}}_{1,2}M^{\rm{s}}_{2,1}} \frac{\rm{Var}([A_1+ A_2]\cdot 1_{A_{\rm{ss}}^*})}{E([A_1+A_2-d_1-d_2]\cdot 1_{A_{\rm{ss}}^*})} \rightarrow \infty, \label{vareass}\\
\frac{\rm{Var}(V_1^{\rm{s}} \cdot 1_{A_{\rm{ss}}})}{ E(V_1^{\rm{s}} \cdot 1_{A_{\rm{ss}}})^2} &\rightarrow \frac{\rm{Var}([A_1+ A_2]\cdot 1_{A_{\rm{ss}}^*})}{E([A_1+A_2-d_1-d_2]\cdot 1_{A_{\rm{ss}}^*})^2}  < \infty, \quad M^{\rm{s}}_{1,2},M^{\rm{s}}_{2,1} \rightarrow 1. \label{vareeass}
\end{align}

Then we have for the expectation and variance of $V_1^{\rm{s}}$ on $\mathbb{R}^+_0 \times \mathbb{R}^+_0$ that
\begin{align}
E(V_1^{\rm{s}})&= E(V_1^{\rm{s}}\cdot 1_{A_{\rm{ss}}})+ E(V_1^{\rm{s}} \cdot 1_{A_{\rm{ss}}^c}) \rightarrow \infty \quad \text{for } M^{\rm{s}}_{1,2},M^{\rm{s}}_{2,1} \rightarrow 1, \label{ev1} 
\end{align}
and 
\begin{align}
\rm{Var}(V_1^{\rm{s}})=& \rm{Var}(V_1^{\rm{s}} \cdot 1_{A_{\rm{ss}}})+  \rm{Var}(V_1^{\rm{s}}\cdot 1_{A_{\rm{ss}}^c})
-2E(V_1^{\rm{s}} \cdot 1_{A_{\rm{ss}}})E(V_1^{\rm{s}}\cdot 1_{A_{\rm{ss}}^c}),  \label{varv1}
\end{align}
where $ \lim_{ M^{\rm{s}}_{1,2},M^{\rm{s}}_{2,1} \rightarrow 1} E(V_1^{\rm{s}} \cdot 1_{A_{\rm{ss}}^c}) < \infty$ and $ \lim_{ M^{\rm{s}}_{1,2},M^{\rm{s}}_{2,1} \rightarrow 1} \rm{Var}(V_1^{\rm{s}} \cdot 1_{A_{\rm{ss}}^c}) < \infty$, since straightforward calculations show that $V_1^{\rm{s}} \cdot 1_{A_{\rm{ss}}^c}<d_1+ \frac{1}{M^{\rm{s}}_{2,1}}d_2$. Thus, \eqref{vareass} and \eqref{varv1} imply
\begin{align}
\rm{Var}(V_1^{\rm{s}})&\,\rightarrow\, \infty,\quad  M^{\rm{s}}_{1,2},M^{\rm{s}}_{2,1} \rightarrow 1, \label{varinfty}
\intertext{and}
\frac{\rm{Var}(V_1^{\rm{s}})}{E(V_1^{\rm{s}})^2} &\thicksim\footnotemark \frac{\rm{Var}(V_1^{\rm{s}}\cdot 1_{A_{\rm{ss}}})}{ E(V_1^{\rm{s}} \cdot 1_{A_{\rm{ss}}})^2}, \quad  M^{\rm{s}}_{1,2},M^{\rm{s}}_{2,1} \rightarrow 1, \label{asympt1}
\end{align}
because all the other terms in \eqref{ev1} and \eqref{varv1} are dominated by the expressions in \eqref{asympt1}, which go to infinity.  
 Hence,   by \eqref{vareeass},
\begin{align}
\frac{\rm{Var}(V_1^{\rm{s}})}{E(V_1^{\rm{s}})^2} & \rightarrow \frac{\rm{Var}([A_1+A_2]\cdot 1_{{A_{\rm{ss}}^*}})}{E([A_1+A_2- d_1-d_2]\cdot 1_{A_{\rm{ss}}^*})^2} < \infty, \quad\quad M^{\rm{s}}_{1,2},M^{\rm{s}}_{2,1} \rightarrow 1. \label{limVarE}
\end{align}
\footnotetext{For functions $f$ and $g$, $f(x) \thicksim g(x), x \rightarrow x_0,\; \Leftrightarrow\; \lim_{x \rightarrow x_0} \frac{f(x)}{g(x)}=1. $} 
 Altogether, by \eqref{ev1} and \eqref{limVarE},
\begin{align*}
&\lim_{M^{\rm{s}}_{1,2},M^{\rm{s}}_{2,1} \rightarrow 1} \tilde{\sigma} =  \ln \left(\lim_{M^{\rm{s}}_{1,2},M^{\rm{s}}_{2,1} \rightarrow 1}\frac{\rm{Var}(V_1^{\rm{s}})}{E(V_1^{\rm{s}})^2}+1 \right)^{0.5} < \infty
\intertext{and}
&\tilde{\mu} =  -\frac{1}{2} \ln \left(\underbrace{\frac{\rm{Var}(V_1^{\rm{s}})}{E(V_1^{\rm{s}})^4}+ E(V_1^{\rm{s}})^{-2}}_{\rightarrow 0}\right)  \rightarrow \infty\;\; \text{for }{M^{\rm{s}}_{1,2},M^{\rm{s}}_{2,1}  \rightarrow 1}.
\end{align*}

\end{proof}

\begin{lemma} \label{existence.Ass}
Let $A_{\rm{ss}}$ be defined as in \eqref{Ass}. Under cross-ownership of equity only, the pointwise limit  of $1_{A_{\rm{ss}}}$ for $M^{\rm{s}}_{1,2},M^{\rm{s}}_{2,1} \rightarrow 1$ exists and is given by $1_{A_{\rm{ss}}^*}$ with
\begin{equation}
A_{\rm{ss}}^*:=\{(a_1,a_2)\geq 0: a_1+a_2 \geq d_1+d_2\}. \label{limass.eq}
\end{equation}
\end{lemma}
\begin{proof}
Under cross-ownership of equity only, the formula of $A_{\rm{ss}}$ reduces to
\begin{equation}
\{(a_1,a_2)\geq 0: a_1+M^{\rm{s}}_{1,2}a_2\geq d_1+M^{\rm{s}}_{1,2}d_2,\; M^{\rm{s}}_{2,1}a_1+a_2 \geq M^{\rm{s}}_{2,1}d_1+d_2\}. \label{ass.eq}
\end{equation}
Let $(M^{\rm{s}}_{1,2,n_1})_{n_1 \in \mathbb{N}}$ and $(M^{\rm{s}}_{2,1,n_2})_{n_2 \in \mathbb{N}}$ be arbitrary, but strictly increasing sequences in $(0,1)$ with limit 1, and let $A_{\rm{ss,n_1,n_2}}$ stand for  $A_{\rm{ss}}$ associated with the $n_1$th and $n_2$th element the above sequences. Then it is easy to see from \eqref{ass.eq} that $A_{\rm{ss},n_1,n_2}\subset A_{\rm{ss},n_1+1,n_2}$, i.e. $A_{\rm{ss},n_1,n_2}$ is strictly increasing in $n_1$.  Similarly, $A_{\rm{ss},n_1,n_2}$ is strictly increasing in $n_2$. Hence, the indicator function of  $A_{\rm{ss},n_1,n_2}$ is pointwise strictly increasing in $n_1$ and $n_2$, and its pointwise limit exists and is a function with values in $\{0,1\}$ only. As such, this limit is of the form $1_A$ for some set $A \subseteq \mathbb{R}^+_0 \times \mathbb{R}^+_0$. \\
 In order to show $A=A_{\rm{ss}}^*$, we first assume $1_A(a_1^*,a_2^*)=1$, i.e. there is an $N \in \mathbb{N}$ such that $(a_1^*,a_2^*) \in A_{\rm{ss},n_1,n_2}$ for all $n_1,n_2 \geq N$, i.e.
\begin{align*}
a_1^*+M^{\rm{s}}_{1,2,n_1}a_2^*&\geq d_1+M^{\rm{s}}_{1,2,n_1}d_2\quad \text{for all } n_1 \geq N,\\
M^{\rm{s}}_{2,1,n_2}a_1^*+a_2^*&\geq M^{\rm{s}}_{2,1,n_2}d_1+d_2 \quad \text{for all } n_2 \geq N.
\end{align*}
In the limit of $n_1,n_2 \rightarrow \infty$, this means that $a_1^*+a_2^*\geq d_1+d_2$, i.e. $A\subseteq A_{\rm{ss}}^*$. \\
Let now  $1_A(a_1^*,a_2^*)=0$, i.e. for all $n_1,n_2 \in \mathbb{N}$,
\begin{align}
a_1^*+M^{\rm{s}}_{1,2,n_1}a_2^*&< d_1+M^{\rm{s}}_{1,2,n_1}d_2, \label{limass.eq1}\\
M^{\rm{s}}_{2,1,n_2}a_1^*+a_2^*&< M^{\rm{s}}_{2,1,n_2}d_1+d_2.\label{limass.eq2}
\end{align}
In the limit of $n_1,n_2 \rightarrow \infty$, we obtain from \eqref{limass.eq1} and \eqref{limass.eq2} that $a_1^*+a_2^*\leq d_1+d_2$. If we had $a_1^*+a_2^*= d_1+d_2$, \eqref{limass.eq1} and \eqref{limass.eq2} would imply
$a_2^* > d_2$ and $a_1^* > d_1$, in contradiction to $a_1^*+a_2^*= d_1+d_2$. Hence, $a_1^*+a_2^*< d_1+d_2$, i.e. $(a_1^*,a_2^*)\not\in A_{\rm{ss}}^*$, and the assertion follows.
\end{proof}

\begin{lemma} \label{existence}
Let $A_{\rm{ss}}$, $A_{\rm{sd}}$, $A_{\rm{ds}}$ and $A_{\rm{dd}}$ be given by \eqref{Ass}--\eqref{Add}. Under cross-ownership of debt only, the pointwise limits of their indicator functions $1_{A_{\rm{ss}}}$, $1_{A_{\rm{sd}}}$, $1_{A_{\rm{ds}}}$ and $1_{A_{\rm{dd}}}$  exist for $M^{\rm{d}}_{1,2},M^{\rm{d}}_{2,1} \rightarrow 1$. 
 \end{lemma}
\begin{proof}
Under cross-ownership of debt only, we have
\begin{align}
A_{\rm{ss}}=\{(a_1,a_2) \geq 0:&\,  a_1  \geq d_1- M^{\rm{d}}_{1,2}d_2,\;  a_2  \geq d_2-M^{\rm{d}}_{2,1}d_1\},\label{defass}\\
A_{\rm{sd}}=\{(a_1,a_2) \geq 0: &\, a_1+M^{\rm{d}}_{1,2}a_2  \geq (1- M^{\rm{d}}_{1,2}M^{\rm{d}}_{2,1})d_1,\;  a_2  < d_2-M^{\rm{d}}_{2,1}d_1\},\label{defasd}\\
A_{\rm{ds}}=\{(a_1,a_2) \geq 0:  &\,a_1  < d_1-M^{\rm{d}}_{1,2}d_2,\;  M^{\rm{d}}_{2,1}a_1+a_2  \geq (1- M^{\rm{d}}_{1,2}M^{\rm{d}}_{2,1})d_2\},\label{defads}\\
A_{\rm{dd}}=\{(a_1,a_2) \geq 0:  &\,a_1+M^{\rm{d}}_{1,2}a_2  < (1- M^{\rm{d}}_{1,2}M^{\rm{d}}_{2,1})d_1,\;  M^{\rm{d}}_{2,1}a_1+a_2  < (1- M^{\rm{d}}_{1,2}M^{\rm{d}}_{2,1})d_2\}.\label{defadd}
\end{align}

Let $(M^{\rm{d}}_{1,2,n_1})_{n_1 \in \mathbb{N}}$ and $(M^{\rm{d}}_{2,1,n_2})_{n_2 \in \mathbb{N}}$ be arbitrary, but strictly increasing sequences in $(0,1)$ with limit 1, and let $A_{\rm{ss,n_1,n_2}}$, $A_{\rm{sd},n_1,n_2}$,  $A_{\rm{ds},n_1,n_2}$ and  $A_{\rm{dd},n_1,n_2}$ stand for the Suzuki areas associated with the $n_1$th and $n_2$th element the above sequences. \\
First, it is easy to see from \eqref{defass} that $A_{\rm{ss},n_1,n_2}\subset A_{\rm{ss},n_1+1,n_2}$, i.e. $A_{\rm{ss},n_1,n_2}$ is strictly increasing in $n_1$.  Similarly, $A_{\rm{ss},n_1,n_2}$ is strictly increasing in $n_2$. Hence, also the sequence of indicator functions $1_{A_{\rm{ss},n_1,n_2}}$ is pointwise strictly increasing in $n_1$ and $n_2$, i.e. $\lim_{n_1,n_2  \rightarrow \infty} 1_{A_{\rm{ss},n_1,n_2}}$ exists.

Next,
\begin{equation}
A_{\rm{sd},n_1,n_2}=\underbrace{\{a_1+M^{\rm{d}}_{1,2,n_1}a_2  \geq (1- M^{\rm{d}}_{1,2,n_1}M^{\rm{d}}_{2,1,n_2})d_1,\}}_{:=A_{\rm{sd},n_1,n_2,1}} \cap \underbrace{\{ a_2  < d_2-M^{\rm{d}}_{2,1,n_2}d_1\}}_{:=A_{\rm{sd},n_2,2}},\label{splitAsd}
\end{equation}
where $A_{\rm{sd},{n_1},{n_2},1}$ increases in both $n_1$ and $n_2$ and $A_{\rm{sd},n_2,2}$ decreases in $n_2$. Hence, the limits of the associated (separate) indicator functions exist, and because of $1_{A_{\rm{sd},n_1,n_2}}=1_{A_{\rm{sd},n_1,n_2,1}} \times 1_{A_{\rm{sd},n_2,2}}$ for all $n_1,n_2 \in \mathbb{N}$ by \eqref{splitAsd}, the limit of $1_{A_{\rm{sd},n_1,n_2}}$ exists as well. \\
Analogously we can write
\begin{equation}
A_{\rm{ds},n_1,n_2}=\underbrace{\{a_1  < d_1-M^{\rm{d}}_{1,2,n_1}d_2\}}_{:=A_{\rm{ds},n_1,1}} \cap \underbrace{\{ M^{\rm{d}}_{2,1,n_2}a_1+a_2  \geq (1- M^{\rm{d}}_{1,2,n_1}M^{\rm{d}}_{2,1,n_1,n_2})d_2\}}_{:=A_{\rm{ds},n_1,n_2,2}},\label{splitAds}
\end{equation}
with $A_{\rm{ds},{n_1},1}$ decreasing in $n_1$ and $A_{\rm{ds},n_1,n_2,2}$ increasing in both $n_1$ and $n_2$. Hence, the limits of the related indicator functions exist, and thus also the limit of $1_{A_{\rm{ds},n_1,n_2}}$, if $n_1,n_2$ converge to infinity.

With regard to $A_{\rm{dd},n_1,n_2}$, we can argue similarly to  $A_{\rm{ss},n_1,n_2}$. $A_{\rm{dd},n_1,n_2}$ is strictly decreasing in $n_1$ and $n_2$, i.e. the limit of the associated indicator function for $n_1, n_2 \rightarrow \infty$ exists.
\end{proof}

\begin{lemma} \label{gebiete-empty}
Let $d_1,d_2>0$. With $A_{\rm{ss}}^*$, $A_{\rm{sd}}^*$, $A_{\rm{ds}}^*$ and $A_{\rm{dd}}^*$ as defined in \eqref{Aij*}, we have
\begin{align}
&A_{\rm{dd}}^*=\{(0,0)\} \quad \text{for all $d_1$, $d_2>0$}, \label{addempty}\\
&A_{\rm{sd}}^*= \emptyset \quad \Leftrightarrow \quad d_2 < d_1, \label{asdempty}\\
&A_{\rm{ds}}^*= \emptyset \quad \Leftrightarrow \quad d_1 <  d_2. \label{adsempty}
\end{align}
If $d_1=d_2$, $A_{\rm{sd}}^*$ and $A_{\rm{ds}}^*$ equal the strictly positive $a_1$-axis and $a_2$-axis, respectively.
\end{lemma}
\begin{proof}
 Let $(M^{\rm{d}}_{1,2,n_1})_{n_1 \in \mathbb{N}}$, $(M^{\rm{d}}_{2,1,n_2})_{n_2 \in \mathbb{N}}$, $A_{\rm{ss,n_1,n_2}}$, $A_{\rm{sd},n_1,n_2}$,  $A_{\rm{ds},n_1,n_2}$ and  $A_{\rm{dd},n_1,n_2}$ be defined as in the proof of Lemma \ref{existence}. \\
Since $(0,0)\in A_{\rm{dd},n_1,n_2}$ for all $n_1,n_2 \in \mathbb{N}$, i.e. $1_{A_{\rm{dd},n_1,n_2}}(0,0)=1$ for all $n_1,n_2 \in \mathbb{N}$, we also have $\lim_{n_1,n_2 \rightarrow \infty} 1_{A_{\rm{dd},n_1,n_2}}(0,0)=1$, i.e. $(0,0) \in A_{\rm{dd}}^*$. Let us now assume $(a_1^*,a_2^*) \in A_{\rm{dd}}^*$ with $a_1^*+a_2^*>0$, w.l.o.g. let $a_1^*>0$. 
Since $A_{\rm{dd},n_1,n_2}$ is strictly decreasing in $n_1,n_2$ (cf. \eqref{defadd}), we have for all $n_1,n_2 \in \mathbb{N}$,
\begin{align}
a_1^*+M^{\rm{d}}_{1,2,n_1}a_2^* &< (1- M^{\rm{d}}_{1,2,n_1}M^{\rm{d}}_{2,1,n_2})d_1, \label{addbed1}
\end{align}
and as the RHS of \eqref{addbed1} converges to 0 if $n_1$ and $n_2$ go to infinity,  \eqref{addempty} follows.  \\
Let us now assume $d_2 < d_1$ and $(a_1^*,a_2^*)\in A_{\rm{sd}}^*$. Then, by \eqref{splitAsd}, 
\begin{align}
a_2^*&  < d_2-M^{\rm{d}}_{2,1,n_2}d_1 \quad  \text{for all $n_2 \in \mathbb{N}$}. \label{asdbed2}
\end{align}
Since the limit of the RHS of \eqref{asdbed2} for $n_2 \rightarrow \infty$ is negative, such an $(a_1^*,a_2^*)\geq 0$ does  not exist. If $d_2\geq d_1$, it is straightforward to see that $A_{\rm{sd}}^*=\{(a_1,a_2)\geq 0: a_1+a_2>0, a_2 \leq d_2-d_1\}$,  and \eqref{asdempty} follows.  In particular, we have for $d_1=d_2$ that   $A_{\rm{sd}}^*=\{(a_1,a_2)\geq 0: a_1>0, a_2 = 0\}$. 
Analogously, one can show \eqref{adsempty} with the help of \eqref{splitAds}, and we obtain for $d_1 \geq d_1$ that $A_{\rm{ds}}^*=\{(a_1,a_2)\geq 0: a_1+a_2>0, a_1 \leq d_1-d_2\}$, and $A_{\rm{ds}}^*=\{(a_1,a_2)\geq 0:  a_1 =0, a_2>0\}$ if $d_1=d_2$.
\end{proof}

\bigskip

\begin{lemma} \label{inequality}
Let $\mu, \tilde{\mu} \in \mathbb{R}$, $\sigma, \tilde{\sigma}, d_2 \in \mathbb{R}^+$, $\sigma >\tilde{\sigma}$, be such that 
\begin{align}
\exp(\tilde{\mu}+0.5 \tilde{\sigma}^2) &= \exp(\mu+0.5 \sigma^2)+d_2, \label{bed1}\\
(\exp(\tilde{\sigma}^2)-1) \exp(2\tilde{\mu}+\tilde{\sigma}^2) &= (\exp(\sigma^2)-1) \exp(2\mu+\sigma^2) \label{bed2},
\end{align}
which exactly corresponds to the definition of $\tilde{\mu}$ and $\tilde{\sigma}$ in \eqref{defv1d} in Section \ref{d1>d2}. Then 
\begin{align} \label{zuzeigen}
\exp(\tilde{\sigma}\mu-\sigma \tilde{\mu}) <
\frac{\left(\frac{\tilde{\sigma}}{\sigma-\tilde{\sigma}} d_2 \right)^{\tilde{\sigma}}}{\left( \frac{\sigma}{\sigma- \tilde{\sigma}} d_2\right)^{\sigma}}.
\end{align}
\end{lemma}
\begin{proof}
First, \eqref{bed1} and \eqref{bed2} imply that
\begin{align*}
\tilde{\mu}&=\mu+0.5 \sigma^2-0.5 \tilde{\sigma}^2+\ln \left(\frac{\sqrt{\exp(\sigma^2)-1}}{\sqrt{\exp(\tilde{\sigma}^2)-1}} \right),\\
\ln(d_2)&=\mu+0.5 \sigma^2+ \ln \left(\frac{\sqrt{\exp(\sigma^2)-1}}{\sqrt{\exp(\tilde{\sigma}^2)-1}}-1 \right).
\end{align*}
Hence,
\begin{alignat}{2} 
\eqref{zuzeigen}  \Leftrightarrow \quad 0>\,&\tilde{\sigma} \mu-\tilde{\mu} \sigma +(\sigma-\tilde{\sigma}) \ln(d_2)-(\sigma-\tilde{\sigma})\ln(\sigma-\tilde{\sigma})-\tilde{\sigma} \ln(\tilde{\sigma})+\sigma \ln(\sigma) \notag\\
\Leftrightarrow \quad 0>\,&-\sigma \ln \left(\frac{\sqrt{\exp(\sigma^2)-1}}{\sqrt{\exp(\tilde{\sigma}^2)-1}} \right)+0.5 \tilde{\sigma}^2 \sigma-0.5 \tilde{\sigma} \sigma^2 -\tilde{\sigma} \ln(\tilde{\sigma})+\sigma \ln(\sigma)  \notag\\
		&\quad+(\sigma-\tilde{\sigma})\ln\left( \frac{\sqrt{\exp(\sigma^2)-1}}{\sqrt{\exp(\tilde{\sigma}^2)-1}}-1\right) -(\sigma-\tilde{\sigma})\ln(\sigma-\tilde{\sigma}) \notag\\
\Leftrightarrow \quad 0>\,&(\sigma-\tilde{\sigma})\left[\ln \left(\frac{\sqrt{\exp(\sigma^2)-1}}{\sqrt{\exp(\tilde{\sigma}^2)-1}}-1 \right) -\ln(\sigma-\tilde{\sigma})-0.5 \sigma \tilde{\sigma} \right]  \notag\\
&\quad-\tilde{\sigma} \ln(\tilde{\sigma})+\sigma \ln \left(\frac{\sigma \sqrt{\exp(\tilde{\sigma}^2)-1}}{\sqrt{\exp(\sigma^2)-1}} \right). \notag
\end{alignat}
Due to $\sigma > \tilde{\sigma}$, it is sufficient for \eqref{zuzeigen} to show 
\begin{align}
 &(\sigma-\tilde{\sigma}) \left[\ln \left( \frac{\sqrt{\exp(\sigma^2)-1}}{\sqrt{\exp(\tilde{\sigma}^2)-1}}-1\right) -\ln(\sigma-\tilde{\sigma}) \right] -\tilde{\sigma} \ln(\tilde{\sigma})+\sigma \ln \left(\frac{\sigma \sqrt{\exp(\tilde{\sigma}^2)-1}}{\sqrt{\exp(\sigma^2)-1}} \right)<0 \notag 
\intertext{or equivalently}
&\begin{aligned}
(\sigma-\tilde{\sigma})\ln \left(\frac{\sqrt{\exp(\sigma^2)-1}-\sqrt{\exp(\tilde{\sigma}^2)-1}}{\sigma- \tilde{\sigma}} \right) -\sigma \ln &\left(\frac{\sqrt{\exp(\sigma^2)-1}}{\sigma} \right)\\
 & < - \tilde{\sigma} \ln \left(\frac{\sqrt{\exp(\tilde{\sigma}^2)-1}}{\tilde{\sigma}} \right). \label{eq34} \end{aligned}
\end{align}
For that, we consider the derivative of the LHS of \eqref{eq34} with respect to $\sigma$:
\begin{align*}
\frac{\partial }{\partial \sigma} & \left[(\sigma-\tilde{\sigma})\ln \left(\frac{\sqrt{\exp(\sigma^2)-1}-\sqrt{\exp(\tilde{\sigma}^2)-1}}{\sigma- \tilde{\sigma}} \right) -\sigma \ln \left(\frac{\sqrt{\exp(\sigma^2)-1}}{\sigma} \right) \right]\\
&=\ln \left(\frac{\sqrt{\exp(\sigma^2)-1}-\sqrt{\exp(\tilde{\sigma}^2)-1}}{\sigma-\tilde{\sigma}} \right) \\
&\quad +\frac{(\sigma-\tilde{\sigma})^2 \left[\frac{\exp(\sigma^2)\sigma}{\sqrt{\exp(\sigma^2)-1}} (\sigma-\tilde{\sigma})-\left(\sqrt{\exp(\sigma^2)-1}-\sqrt{\exp(\tilde{\sigma}^2)-1}\right) \right]}{\left(\sqrt{\exp(\sigma^2)-1}-\sqrt{\exp(\tilde{\sigma}^2)-1}\right)(\sigma-\tilde{\sigma})^2}\\
&\quad -\ln\left(\frac{\sqrt{\exp(\sigma^2)-1}}{\sigma}\right)- \frac{\sigma^2 \left[\frac{\exp(\sigma^2)\sigma}{\sqrt{\exp(\sigma^2)-1}}-\sqrt{\exp(\sigma^2)-1} \right]}{\sqrt{\exp(\sigma^2)-1}\, \sigma^2}\\
&=\ln \left(\frac{\sigma}{\sqrt{\exp(\sigma^2)-1}}\right)-\ln\left(\frac{\sigma-\tilde{\sigma}}{\sqrt{\exp(\sigma^2)-1}-\sqrt{\exp(\tilde{\sigma}^2)-1}}\right)\\
&\quad+\frac{\exp(\sigma^2)\sigma}{\sqrt{\exp(\sigma^2)-1}}\left[\frac{\sigma-\tilde{\sigma}}{\sqrt{\exp(\sigma^2)-1}-\sqrt{\exp(\tilde{\sigma}^2)-1}}-\frac{\sigma}{\sqrt{\exp(\sigma^2)-1}} \right].
\end{align*}
Because of $\frac{\sigma-\tilde{\sigma}}{\sqrt{\exp(\sigma^2)-1}-\sqrt{\exp(\tilde{\sigma}^2)-1}}-\frac{\sigma}{\sqrt{\exp(\sigma^2)-1}} <0$, this derivative is negative if and only if
\begin{align}
\frac{\ln \left(\frac{\sigma}{\sqrt{\exp(\sigma^2)-1}}\right)-\ln \left(\frac{\sigma-\tilde{\sigma}}{\sqrt{\exp(\sigma^2)-1}-\sqrt{\exp(\tilde{\sigma}^2)-1}}\right)}{\frac{\sigma}{\sqrt{\exp(\sigma^2)-1}}-\frac{\sigma-\tilde{\sigma}}{\sqrt{\exp(\sigma^2)-1}-\sqrt{\exp(\tilde{\sigma}^2)-1}}} < \frac{\exp(\sigma^2)\sigma}{\sqrt{\exp(\sigma^2)-1}}. \label{eq5}
\end{align}
Since the LHS of \eqref{eq5} can be interpreted as the difference quotient of the concave logarithmic function  in $x_0=\frac{\sigma}{\sqrt{\exp(\sigma^2)-1}}$ and $x=\frac{\sigma-\tilde{\sigma}}{\sqrt{\exp(\sigma^2)-1}-\sqrt{\exp(\tilde{\sigma}^2)-1}}$, the LHS of \eqref{eq5} is strictly decreasing in $x$ and thus strictly increasing in $\tilde{\sigma}$. From 
\begin{align*}
\lim_{\tilde{\sigma} \nearrow \sigma} \frac{\sigma-\tilde{\sigma}}{\sqrt{\exp(\sigma^2)-1}-\sqrt{\exp(\tilde{\sigma}^2)-1}} &=\left(\lim_{\tilde{\sigma} \nearrow \sigma} \frac{\sqrt{\exp(\sigma^2)-1}-\sqrt{\exp(\tilde{\sigma}^2)-1}}{\sigma-\tilde{\sigma}}\right)^{-1}\\
&=\left(\frac{\partial}{\partial \sigma} \sqrt{\exp(\sigma^2)-1}\right)^{-1}=\frac{\sqrt{\exp(\sigma^2)-1}}{\exp(\sigma^2)\sigma}
\end{align*}
it follows that the LHS of \eqref{eq5} is smaller than 
\[
\frac{\ln\left(\frac{\sigma}{\sqrt{\exp(\sigma^2)-1}}\right)-\ln\left(\frac{\sqrt{\exp(\sigma^2)-1}}{\exp(\sigma^2)\sigma}\right)}{\frac{\sigma}{\sqrt{\exp(\sigma^2)-1}}-\frac{\sqrt{\exp(\sigma^2)-1}}{\exp(\sigma^2)\sigma}} < \frac{\exp(\sigma^2)\sigma}{\sqrt{\exp(\sigma^2)-1}},
\]
where the last inequality follows from straightforward calculations and the fact that $\ln(x)<x-1 \;\forall \; x>0$. Thus, \eqref{eq5} is met for all $\tilde{\sigma}<\sigma$, 
i.e. the LHS of \eqref{eq34} is strictly decreasing in $\sigma$ for all $\tilde{\sigma}>0$. Hence,  for \eqref{eq34} it only remains to show that \eqref{eq34} holds in the limit of $\sigma \searrow \tilde{\sigma}$.
Because of 
\[
\left| \lim_{\sigma \searrow \tilde{\sigma}} \frac{\sqrt{\exp(\sigma^2)-1}-\sqrt{\exp(\tilde{\sigma}^2)-1}}{\sigma- \tilde{\sigma}} \right| = \left|\frac{\partial}{\partial \tilde{\sigma}} \sqrt{\exp(\tilde{\sigma}^2)-1} \right| = \left|\frac{\exp(\tilde{\sigma}^2)\tilde{\sigma}}{\sqrt{\exp(\tilde{\sigma}^2)-1}} \right|<\infty,
\]
 we  have
\begin{align*}
\lim_{\sigma \searrow \tilde{\sigma}} &\left[(\sigma-\tilde{\sigma})\ln \left(\frac{\sqrt{\exp(\sigma^2)-1}-\sqrt{\exp(\tilde{\sigma}^2)-1}}{\sigma- \tilde{\sigma}} \right) -\sigma \ln \left(\frac{\sqrt{\exp(\sigma^2)-1}}{\sigma} \right) \right] \\
&= - \tilde{\sigma} \ln \left(\frac{\sqrt{\exp(\tilde{\sigma}^2)-1}}{\tilde{\sigma}} \right),
\end{align*}
i.e. \eqref{eq34}  and \eqref{zuzeigen} follow.
\end{proof}

\end{document}